\documentclass{llncs}

\usepackage{microtype}

\usepackage{epsfig, amssymb, amsmath}
\usepackage{pstricks, pst-node, pst-tree, color}
\usepackage{amsfonts}

\def\ZP{{\mathbb{Z}^{+}}}
\title{Approximating the Revenue Maximization Problem with Sharp Demands
}

\author{Vittorio Bil\`o\inst{1} \and Michele Flammini\inst{2} \and Gianpiero Monaco \inst{2}}
\institute{{Department of Mathematics and Physics ``Ennio De
Giorgi", University of Salento\\
Provinciale Lecce-Arnesano, P.O. Box 193, 73100 Lecce, Italy\\ {\tt vittorio.bilo@unisalento.it}}\\
\and{Department of Information Engineering Computer Science and Mathematics, University of L'Aquila, Via Vetoio, Coppito, 67100 L'Aquila, Italy\\
{\tt \{flammini,gianpiero.monaco\}@di.univaq.it}}}

\begin{document}
\maketitle

\begin{abstract}
We consider the revenue maximization problem with sharp
multi-demand, in which $m$ indivisible items have to be sold to $n$
potential buyers. Each buyer $i$ is interested in getting exactly
$d_i$ items, and each item $j$ gives a benefit $v_{ij}$ to buyer
$i$. We distinguish between unrelated and related valuations. In the
former case, the benefit $v_{ij}$ is completely arbitrary, while, in
the latter, each item $j$ has a quality $q_j$, each buyer $i$ has a
value $v_i$ and the benefit $v_{ij}$ is defined as the product $v_i
q_j$. The problem asks to determine a price for each item and an
allocation of bundles of items to buyers with the aim of maximizing
the total revenue, that is, the sum of the prices of all the sold
items. The allocation must be envy-free, that is, each buyer must be
happy with her assigned bundle and cannot improve her utility. We
first prove that, for related valuations, the problem cannot be
approximated to a factor $O(m^{1-\epsilon})$, for any $\epsilon>0$,
unless {\sf P} = {\sf NP} and that such result is asymptotically
tight. In fact we provide a simple $m$-approximation algorithm even
for unrelated valuations. We then focus on an interesting subclass
of "proper" instances, that do not contain buyers a priori known not
being able to receive any item. For such instances, we design an
interesting $2$-approximation algorithm and show that no
$(2-\epsilon)$-approximation is possible for any $0<\epsilon\leq 1$,
unless {\sf P} $=$ {\sf NP}. We observe that it is possible to
efficiently check if an instance is proper, and if discarding
useless buyers is allowed,  an instance can be made proper in
polynomial time, without worsening the value of its optimal
solution.
\end{abstract}

\section{Introduction}
A major decisional process in many business activities concerns whom
to sell products (or services) to and at what price, with the goal
of maximizing the total revenue. On the other hand, consumers would
like to buy at the best possible prices and experience fair sale
criteria.

In this work, we address such a problem from a computational point
of view, considering a two-sided market in which the supply side
consists of $m$ indivisible items and the demand one is populated by
$n$ potential buyers (in the following also called consumers or
customers), where each buyer $i$ has a demand $d_i$ (the number of
items that $i$ requests) and valuations $v_{ij}$ representing the
benefit $i$ gets when owing item $j$. As several papers on this
topic (see for instance \cite{DK88,MZ88,GHKKKM05,CDGZ12,FFLS12}),
we assume that, by means of market research or interaction with the
consumers, the seller knows each customer's valuation for each item.

The seller sets up a price $p_j$ for each item $j$
and assigns (i.e., sells) bundle of items to buyers with the aim of maximizing
her revenue, that is the sum of the prices of all the sold items. When a consumer is assigned (i.e., buys) a set
of items, her utility is the difference between the total valuation
of the items she gets (valuations being additive) and the
purchase price.

The sets of the sold items, the purchasing customers and their
purchase prices are completely determined by the allocation of
bundles of items to customers unilaterally decided by the seller.
Nevertheless, we require such an allocation to meet two basic
fairness constraints: (i) each customer $i$ is allocated at most one
bundle not exceeding her demand $d_i$ and providing her a
non-negative utility, otherwise she would not buy the bundle; (ii),
the allocation must be envy-free \cite{W54}, i.e., each customer $i$
does not prefer any subset of $d_i$ items different from the bundle
she is assigned.

The envy-freeness notion adopted in this paper is the typical one of
pricing problems. Anyway, in the literature there also exist weaker
forms usually applied in fair division settings (see for instance
\cite{F67}) where, basically, no buyer wants to switch her
allocation with that of another buyer, without combining different
bundles. Notice that in our scenario a trivial envy-free solution
always exists that lets $p_j=\infty$ for each item $j$ and does not
assign any item to any buyer.

Many papers (see the {\em Related Work} section for a detailed reference list) considered
the {\em unit demand case} in which $d_i=1$ for each consumer $i$. Arguably, the
{\em multi-demand case}, where $d_i\geq 1$ for each consumer $i$, is more general and finds much more applicability.
To this aim, we can identify two main multi-demand schemes.
The first one is the {\em relaxed multi-demand model}, where each buyer $i$ requests at
most $d_i\geq 1$ items, and the second one is the {\em sharp multi-demand
model}, where each buyer $i$ requests exactly $d_i\geq 1$ items and,
therefore, a bundle of size less than $d_i$ has no value for buyer $i$.

For relaxed multi-demand models, a standard technique can reduce the
problem to the unit demand case in the following way: each buyer $i$
with demand $d_i$ is replaced by $d_i$ copies of buyer $i$, each
requesting a single item. However, such a trick does not apply to
the sharp demand model. Moreover, as also pointed out in
\cite{CDGZ12}, the sharp multi-demand model exhibits a property that
unit demand and relaxed multi-demand ones do not posses. In fact,
while in the latter model any envy-free pricing is such that the
price $p_j$ is always at most the value of $v_{ij}$, in the sharp
demand model, a buyer $i$ may pay an item $j$ more than her own
valuation for that item, i.e., $p_j>v_{ij}$ and compensate her loss
with profits from the other items she gets (see section 3.1 of
\cite{CDGZ12}). Such a property, also called {\em overpricing},
clearly adds an extra challenge to find an optimal revenue.

The sharp demand model is quite natural in several settings. Consider, for
instance, a scenario in which a public organization has the need of
buying a fixed quantity of items in order to reach a specific
purpose (i.e. locations for offices, cars for services, bandwidth, storage, or whatever
else), where each item might have a different valuation for the
organization because of its size, reliability, position, etc. Yet,
suppose a user wants to store on a remote server a file of a given
size $s$ and there is a memory storage vendor that sells slots of fixed size $c$, where each cell might have different
features depending on the server location and speed and then yielding different valuations for the user. In this
case, a number of items smaller than $\left\lceil\frac{s}{c}\right\rceil$ has no
value for the user. Similar scenarios also apply to cloud computing.
In \cite{CDGZ12}, the authors used the following applications for the sharp
multi-demand model. In TV (or radio) advertising \cite{NBCFGMRSTVZ09},
advertisers may request different lengths of advertising slots for
their ads programs. In banner (or newspaper) advertising,
advertisers may request different sizes or areas for their displayed
ads, which may be decomposed into a number of base units. Also,
consider a scenario in which advertisers choose to display their
advertisement using medias (video, audio, animation)
\cite{BCI98,R09} that would usually need a fixed number of positions,
while text ads would need only one position each. An example of
formulation sponsored search using sharp multi-demands can be found in
\cite{DSYZ10}. Other results concerning the sharp multi-demand model in the Bayesian setting can be found in \cite{DGTZ12}.

\noindent{\bf Related Work.} Pricing problems have been intensively
studied in the literature, see e.g.,
\cite{S82,SS79,M01a,M01b,R03,AFMZ04,GRV06} just to cite a few, both
in the case in which the consumers' preferences are unknown
(mechanism design \cite{V61,C71}) and in the case of full
information that we consider in this paper.
\iffalse A common goal
for economically motivated algorithms and mechanisms is economic
efficiency, that is obtaining an outcome which maximizes the sum of
the utilities of all the participants. Such a problem has been
largely investigated from a mechanism design point of view. For
instance, in the unit demand case, the Vickrey-Clarke-Groves (VCG)
mechanism \cite{V61,C71} solves the problem, and is both truthful
and envy-free. Unfortunately, computing the VCG mechanism is
NP-hard.\fi
 In fact, our interest here is in maximizing the seller's
profit assuming that consumers' preferences are gathered through
market research or conjoint analysis
\cite{DK88,MZ88,GHKKKM05,CDGZ12,FFLS12}. From an algorithmic point
of view, \cite{GHKKKM05} is the first paper dealing with the problem
of computing the envy-free pricing of maximum revenue. The authors
considered the unit demand case for which they gave an $O(\log
n)$-approximation algorithm and showed that computing an optimal
envy-free pricing is APX-hard. Briest \cite{B08} showed that, under
reasonable complexity assumptions,
%i.e., assuming specific
%hardness of the balanced bipartite independent set problem in
%constant degree graphs or hardness of refuting random 3CNF formulas,
the revenue maximization problem in the unit demand model cannot be approximated within
$O(\log^{\varepsilon} n )$ for some $\varepsilon > 0$.
The subcase in which every buyer positively evaluates at most
two items has been studied in \cite{CD10}. The authors proved that the problem
is solvable in polynomial time and it becomes NP-hard if some buyer gets
interested in at least three items.

For the multi-demand model, Chen et. al. \cite{CGV11} gave an
$O(\log D)$ approximation algorithm when there is a metric space
behind all items, where $D$ is the maximum demand, and Briest
\cite{B08} showed that the problem is hard to approximate within a
ratio of $O(n^{\varepsilon})$ for some $\varepsilon > 0$.

To the best of our knowledge, \cite{CDGZ12} is the first paper
explicitly dealing with the sharp multi-demand model. The authors
considered a particular valuation scheme (also used in
\cite{EOS07,V07} for keywords advertising scenarios) where each item
$j$ has a parameter $q_j$ measuring the quality of the item and each
buyer $i$ has a value $v_i$ representing the benefit that $i$ gets
when owing an item of unit quality. Thus, the benefit that $i$
obtains from item $j$ is given by $v_iq_j$. For such a problem, the
authors proved that computing the envy-free pricing of maximum
revenue is NP-hard. Moreover, they showed that if the demand of each
buyer is bounded by a constant, the problem becomes solvable in
polynomial time. We remark that this valuation scheme is a special
case of the one in which the valuations $v_{ij}$ are completely
arbitrary and given as an input of the problem. Throughout the
paper, we will refer to the former scheme as to {\em related
valuations} and to the latter as to {\em unrelated valuations}.
Recently \cite{DGSTZ13} considered the sharp multi-demand model with
the additional constraint in which items are arranged as a sequence
and buyers want items that are consecutive in the sequence.

Finally \cite{FFLS12} studied the pricing problem in the case in
which buyers have a budget, but no demand constraints. The authors
considered a special case of related valuations in which all
qualities are equal to $1$ (i.e., $q_j=1$ for each item $j$). They
proved that the problem is still NP-hard and provided a
$2$-approximation algorithm. Such algorithm assigns the same price
to all the sold items.

Many of the papers listed above deal with the case of limited
supply. Another stream of research considers unlimited supply, that
is, the scenario in which each item $j$ exists in $e_j$ copies and
it is explicitly allowed that $e_j=\infty$. The limited supply
setting seems generally more difficult than the unlimited supply
one. In this paper we consider the limited supply setting.
Interesting results for unlimited supply can be found in
\cite{GHKKKM05,DFHS08,BBM08}.

\noindent{\bf Our Contribution.}
We consider the revenue maximization problem with sharp multi-demand
and limited supply. We first prove that, for related valuations, the
problem cannot be approximated to a factor $O(m^{1-\epsilon})$, for
any $\epsilon>0$, unless {\sf P} = {\sf NP} and that such result is
asymptotically tight. In fact we provide a simple $m$-approximation
algorithm even for unrelated valuations.

Our inapproximability proof relies on the presence of some buyers
not being able to receive any bundle of items in any envy-free
outcome. Thus, it becomes natural to ask oneself what happens for
instances of the problem, that we call {\em proper}, where no such
pathological buyers exist. For proper instances, we design an
interesting $2$-approximation algorithm and show that the problem
cannot be approximated to a factor $2-\epsilon$ for any
$0<\epsilon\leq 1$ unless {\sf P} $=$ {\sf NP}. Therefore, also in
this subcase, our results are tight. We remark that it is possible
to efficiently decide whether an instance is proper. Moreover, if
discarding useless buyers is allowed, an instance can be made proper
in polynomial time, without worsening the value of its optimal
solution.

%%%%%%%%%%%     PARTE VITTORIO   %%%%%%%%%%

\section{Model and Preliminaries}
In the {\em Revenue Maximization Problem with Sharp Multi-Demands}
({\sf RMPSD}) investigated in this paper, we are given a {\em
market} made up of a set $M=\{1,2,\ldots,m\}$ of {\em items} and a
set $N=\{1,2,\ldots,n\}$ of {\em buyers}. Each item $j\in M$ has
unit supply (i.e., only one available copy). We consider both
unrelated and related valuations. In the former each buyers $i$ has
valuations $v_{ij}$ representing the benefit $i$ gets when owing
item $j$. In the latter each item is characterized by a {\em
quality} (or desirability) $q_j>0$, while each buyer $i\in N$ has a
{\em value} $v_i>0$, measuring the benefit that she gets when
receiving a unit of quality, thus, the valuation that buyer $i$ has
for item $j$ is $v_{ij}=v_i q_j$. We notice that related is a
special case of unrelated valuations. Throughout the paper, when not
explicitly indicated, we refer to related valuations. Finally each
buyer $i$ has a {\em demand} $d_i\in\ZP$, which specifies the exact
number of items she wants to get. In the following we assume items
and bidders ordered in non-increasing order, that is, $v_i\geq
v_{i'}$ for $i<i'$ and $q_j\geq q_{j'}$ for $j<j'$.

An {\em allocation vector} is an $n$-tuple ${\bf
X}=(X_1,\ldots,X_n)$, where $X_i\subseteq M$, with
$|X_i|\in\{0,d_i\}$, $\sum_{i\in N}|X_i|\leq m$ and $X_i\cap
X_{i'}=\emptyset$ for each $i\neq i'\in N$, is the set of items sold
to buyer $i$. A {\em price vector} is an $m$-tuple ${\bf
p}=(p_1,\ldots,p_m)$, where $p_j>0$ is the price of item $j$. An
{\em outcome} of the market is a pair $({\bf X},{\bf p})$.

Given an outcome $({\bf X},{\bf p})$, we denote with $u_{ij}({\bf
p})=v_{ij}-p_j$ the utility that buyer $i$ gets when she is sold
item $j$ and with $u_i({\bf X},{\bf p})=\sum_{j\in X_i}u_{ij}({\bf
p})$ the overall utility of buyer $i$ in $({\bf X},{\bf p})$. When
the outcome (or the price vector) is clear from the context, we
simply write $u_i$ and $u_{ij}$. An outcome $({\bf X},{\bf p})$ is
{\em feasible} if $u_i\geq 0$ for each $i\in N$.

We denote with $M({\bf X})=\bigcup_{i\in N}X_i$ the set of items
sold to some buyer according to the allocation vector $\bf X$. We
say that a buyer $i$ is a {\em winner} if $X_i\neq\emptyset$ and we
denote with $W({\bf X})$ the set of all the winners in $\bf X$. For
an item $j\in M({\bf X})$, we denote with $b_{\bf X}(j)$ the buyer
$i\in W({\bf X})$ such that $j\in X_i$, while, for an item $j\notin
M({\bf X})$, we define $b_{\bf X}(j)=0$. Moreover, for a winner
$i\in W({\bf X})$, we denote with $f_{\bf X}(i)=\min\{j\in M:j\in
X_i\}$ the best-quality item in $X_i$. Also in this case, when the
allocation vector is clear from the context, we simply write $b(j)$
and $f(i)$. Finally, we denote with $\beta({\bf X})=\max\{i\in
N:i\in W({\bf X})\}$ the maximum index of a winner in $\bf X$. An
allocation vector $\bf X$ is {\em monotone} if $\min_{j\in
X_i}\{q_j\}\geq q_{f(i')}$ for each $i,i'\in W({\bf X})$ with
$v_i>v_{i'}$, that is, all the items of $i$ are of quality greater
of equal to the one of all the items of $i'$.

\begin{definition}
A feasible outcome $({\bf X},{\bf p})$ is an envy-free outcome if,
for each buyer $i\in N$, $u_i\geq\sum_{j\in T}u_{ij}$ for each
$T\subseteq M$ of cardinality $d_i$.
\end{definition}

Notice that, by definition, an outcome $({\bf X},{\bf p})$ is envy-free if and only if the following three conditions holds:
\begin{enumerate}
\item $u_i\geq 0$ for each $i\in N$,
\item $u_{ij}\geq u_{ij'}$ for each $i\in W({\bf X})$, $j\in X_i$ and $j'\notin X_i$,
\item $\sum_{j\in T}u_{ij}\leq 0$ for each $i\notin W({\bf X})$ and $T\subseteq M$ of cardinality $d_i$.
\end{enumerate}
Note also that, as already remarked, envy-free solutions always
exist, since the outcome $({\bf X},{\bf p})$ such that
$X_i=\emptyset$ for each $i\in N$ and $p_j=\infty$ for each $j\in M$
is envy-free. Moreover, deciding whether an outcome is envy-free can
be done in polynomial time.

By the definition of envy-freeness, if $i \in W(X)$ is a winner,
then all the buyers $i'$ with $v_{i'} > v_i$ and $d_{i'} \leq d_i$
must be winners as well, otherwise $i'$ would envy a subset of the
bundle assigned to $i$. This motivates the following definition,
which restricts to instances not containing buyers not being a
priori able to receive any item ({\em useless buyers}).

\begin{definition}
An instance $I$ is proper if, for each buyer $i\in N$, it holds
$d_i+\sum_{i' | v_{i'} > v_i,d_{i'} \leq d_i} d_{i'} \leq m$.
\end{definition}

The (market) {\em revenue} generated by an outcome $({\bf X},{\bf
p})$ is defined as $rev({\bf X},{\bf p})=\sum_{j\in M({\bf X})}p_j$.
{\sf RMPSD} asks for the determination of an envy-free outcome of
maximum revenue. We observe that it is possible to efficiently check
if an instance is proper, and if discarding useless buyers is
allowed,  an instance can be made proper in polynomial time, without
worsening the value of its optimal solution. An instance of the {\sf
RMPSD} problem can be modeled as a triple $({\bf V},{\bf D},{\bf
Q})$, where ${\bf V}=(v_1,\ldots,v_n)$ and ${\bf
D}=(d_1,\ldots,d_n)$ are the vectors of buyers' values and demands,
while ${\bf Q}=(q_1,\ldots,q_m)$ is the vector of item qualities. We
conclude this section with three lemmas describing some properties
that need to be satisfied by any envy-free outcome.

\begin{lemma}[\cite{CDGZ12}]\label{orderedqualities}
If an outcome $({\bf X},{\bf p})$ is envy-free, then $\bf X$ is monotone.
\end{lemma}

\begin{proof}
Let $({\bf X},{\bf p})$ be an envy-free outcome and assume, for the
sake of contradiction, that $\bf X$ is not monotone, i.e., that
there exist two buyers $i,i'\in W({\bf X})$ with $v_i>v_i'$ and two
items $j\in X_i$ and $j'\in X_{i'}$ such that $q_j<q_{j'}$. By the
envy-freeness of $({\bf X},{\bf p})$, it holds $u_{ij}\geq u_{ij'}$
which implies $p_j-p_{j'}\leq v_i(q_j-q_{j'})$ and $u_{i'j'}\geq
u_{i'j}$ which implies $p_j-p_{j'}\geq v_{i'}(q_j-q_{j'})$. By
dividing both inequalities by $q_j-q_{j'}<0$, we get
$v_{i'}\geq\frac{p_j-p_{j'}}{q_j-q_{j'}}$ and
$v_i\leq\frac{p_j-p_{j'}}{q_j-q_{j'}}$ which implies $v_i\leq
v_{i'}$, a contradiction.\qed
\end{proof}

Given an outcome $({\bf X},{\bf p})$, an item $j\in X_i$ is {\em overpriced} if $u_{ij}<0$.

\begin{lemma}[\cite{CDGZ12}]\label{nooverprice}
Let $({\bf X},{\bf p})$ be an envy-free outcome. For each overpriced item $j'\in M({\bf X})$, it holds $b(j')=\beta({\bf X})$.
\end{lemma}

\begin{proof}
Let $({\bf X},{\bf p})$ be an envy-free outcome and assume, for the
sake of contradiction, that there exists an overpriced item $j\in
X_i$ with $i<\beta({\bf X})$. Hence, $u_{ij}=v_{ij} -p_j<0$. Since
$\beta({\bf X})\in W({\bf X})$ and $({\bf X},{\bf p})$ is feasible,
it holds $u_{\beta({\bf X})}\geq 0$ which implies that there exists
an item $j'\in X_{\beta({\bf X})}$ such that $u_{\beta({\bf X})
j'}=v_{\beta({\bf X}) j'}-p_{j'}\geq 0$. Moreover, by the
envy-freeness of $({\bf X},{\bf p})$, it also holds $u_{ij}\geq
u_{ij'}$. By using $v_i\geq v_{\beta({\bf X})}$, we get $u_{ij}\geq
u_{ij'}=v_i q_{j'}-p_{j'}\geq v_{\beta({\bf
X})}q_{j'}-p_{j'}=u_{\beta({\bf X}) j'}\geq 0$, which contradicts
the assumption that $u_{ij}<0$.\qed
\end{proof}

The following lemma establishes that, if a buyer $i$ is not a
winner, then the total number of items assigned to buyers with
valuation strictly smaller than $v_i$ is less than $d_i$.

\begin{lemma}\label{noholes}
Let $({\bf X},{\bf p})$ be an envy-free outcome. For each buyer $i$ such that $i\notin W({\bf X})$, it holds $d_i>\sum_{k>i:k\in W({\bf X})}d_{k}$.
\end{lemma}

\begin{proof}
Let $({\bf X},{\bf p})$ be an envy-free outcome and let $i,i'$ be
two buyers such that $v_i>v_{i'}$, $i\notin W({\bf X})$ and $i'\in
W({\bf X})$. Assume, for the sake of contradiction, that $d_i\leq
\sum_{k>i:k\in W({\bf X})}d_{k}$. This implies that there exists
$T\subseteq M$, of cardinality $d_i$, such that all items $j\in T$
are assigned to buyers with values of at most $v_i$ and at least one
item $j'\in T$ is assigned to buyer $i'$. Moreover, since $u_k\geq
0$ for each $k\in W({\bf X})$ by the feasibility of $({\bf X},{\bf
p})$, there exists one such $T$ for which $u_{i'j'}+\sum_{j\in
T\setminus\{j'\}}u_{b(j)j}\geq 0$. Hence, we obtain $$\sum_{j\in
T}u_{ij}=u_{ij'}+\sum_{j\in
T\setminus\{j'\}}u_{ij}>u_{i'j'}+\sum_{j\in
T\setminus\{j'\}}u_{b(j)j}\geq 0,$$ where the strict inequality
follows from the fact that $v_i>v_{i'}$ and $v_i\geq v_{b(j)}$ for
each $j\in T\setminus\{j'\}$. Thus, since there exists a set of
items $T$ of cardinality $d_i$ such that $\sum_{j\in T}u_{ij}>0$, it
follows that $({\bf X},{\bf p})$ is not envy-free, a
contradiction.\qed
\end{proof}

\section{A Pricing Scheme for Monotone Allocation Vectors}
Since we are interested only in envy-free outcomes, by Lemma
\ref{orderedqualities}, in the following we will implicitly assume
that any considered allocation vector is monotone.

We call {\em pricing scheme} a function which, given an allocation
vector $\bf X$, returns a price vector. In this section, we propose
a pricing scheme for allocation vectors which will be at the basis
of our approximability and inapproximability results. For the sake of readability,
in describing the following pricing function, given $\bf X$,
we assume a re-ordering of the buyers in such a way that all the winners appear first, still
in non-increasing order of $v_i$.

\noindent For an allocation vector $\bf X$, define the price vector $\widetilde{{\bf p}}$ such that, for each $j\in M$,
\begin{displaymath}
\widetilde{p}_j=\left\{
\begin{array}{ll}
\infty & \textrm{  if }b(j)=0,\\
v_{b(j)} q_j-\displaystyle\sum_{k=b(j)+1}^{\beta({\bf X})}\left((v_{k-1}-v_k)q_{f(k)}\right) & \textrm{ otherwise}.
\end{array}\right.
\end{displaymath}
Quite interestingly, such a scheme resembles one presented \cite{HY11}.
Next lemma shows that $\widetilde{{\bf p}}$ is indeed a price vector.

\begin{lemma}\label{prezzinonnegativi}
For each $j\in M$, it holds $\widetilde{p}_j>0$.
\end{lemma}

\begin{proof}
Clearly, the claim holds for each $j$ such that $b(j)\in\{0,\beta({\bf X})\}$. For each $j$ such that $0<b(j)<\beta({\bf X})$, it holds
\begin{eqnarray*}
\widetilde{p}_j & = & v_{b(j)} q_j-\sum_{k=b(j)+1}^{\beta({\bf X})}\left((v_{k-1}-v_k)q_{f(k)}\right)\\
& = & v_{b(j)}(q_j-q_{f(b(j)+1)})+\sum_{k=b(j)+1}^{\beta({\bf X})-1}\left((q_{f(k)}-q_{f(k+1)})v_k\right)+v_{\beta({\bf X})} q_{f({\beta({\bf X})})}\\
& > & 0,
\end{eqnarray*}
where the inequality holds since $v_{\beta({\bf X})} q_{f({\beta({\bf X})})}>0$ and all the other terms are non-negative since $\bf X$ is monotone.\qed
\end{proof}

We continue by showing the following important property, closely
related to the notion of envy-freeness, possessed by the outcome
$({\bf X},\widetilde{{\bf p}})$ for each allocation vector $\bf X$.

\begin{lemma}\label{prezzatura}
For each allocation vector $\bf X$, the outcome $({\bf
X},\widetilde{{\bf p}})$ is feasible and, for each winner $i\in
W({\bf X})$, $u_i\geq\sum_{j\in T}u_{ij}$ for each $T\subseteq M$ of
cardinality $d_i$. Thus, the allocation is envy-free for the subset
of the winners buyers.
\end{lemma}

\begin{proof}
Given an allocation vector $\bf X$, consider a winner $i\in W({\bf
X})$. If $i$ is the only winner in $W({\bf X})$, it immediately
follows that $u_i\geq\sum_{j\in T}u_{ij}$ for each $T\subseteq M$ of
cardinality $d_i$ since items not assigned to $i$ have infinite
price. We prove this claim for the case in which $|W({\bf X})|>1$ by
showing that, for each $j,j'\in M$ such that $j\in X_i$ and
$j'\notin X_i$, it holds $u_{ij}\geq u_{ij'}$.

To this aim, consider an item $j'$ such that $0<b(j')<i=b(j)$ (whenever it exists). It holds
\begin{eqnarray*}
u_{ij}-u_{ij'} & = & v_i q_j - \widetilde{p}_j - v_i q_{j'} + \widetilde{p}_{j'}\\
& = & v_i q_{j}- v_{b(j)} q_{j}+\sum_{k=b(j)+1}^{\beta({\bf X})}\left((v_{k-1}-v_k)q_{f(k)}\right)\\
& & -v_i q_{j'}+v_{b(j')}q_{j'}-\sum_{k=b(j')+1}^{\beta({\bf X)}}\left((v_{k-1}-v_k)q_{f(k)}\right)\\
& = & v_{b(j')}q_{j'}- v_i q_{j'}+\sum_{k=i+1}^{\beta({\bf X)}}\left((v_{k-1}-v_k)q_{f(k)}\right)-\sum_{k=b(j')+1}^{\beta({\bf X)}}\left((v_{k-1}-v_k)q_{f(k)}\right)\\
& = & (v_{b(j')}-v_i)q_{j'}-\sum_{k=b(j')+1}^{i}\left((v_{k-1}-v_k)q_{f(k)}\right)\\
& \geq & (v_{b(j')}-v_i)q_{j'}-\sum_{k=b(j')+1}^{i}\left((v_{k-1}-v_k)q_{j'}\right)\\
& = & (v_{b(j')}-v_i)q_{j'}-(v_{b(j')}-v_i)q_{j'}\\
& = & 0,
\end{eqnarray*}
where the second equality comes from $i=b(j)$ and the inequality follows from the monotonicity of $\bf X$.

Now consider an item $j'$ such that $b(j')>i=b(j)$ (whenever it exists). Similarly as above, it holds
\begin{eqnarray*}
u_{ij}-u_{ij'} & = & v_i q_j - \widetilde{p}_j - v_i q_{j'} + \widetilde{p}_{j'}\\
& = & v_{b(j')}q_{j'}- v_i q_{j'}+\sum_{k=i+1}^{\beta({\bf X})}\left((v_{k-1}-v_k)q_{f(k)}\right)-\sum_{k=b(j')+1}^{\beta({\bf X})}\left((v_{k-1}-v_k)q_{f(k)}\right)\\
& = & (v_{b(j')}-v_i)q_{j'}+\sum_{k=i+1}^{b(j')}\left((v_{k-1}-v_k)q_{f(k)}\right)\\
& \geq & (v_{b(j')}-v_i)q_{j'}+\sum_{k=i+1}^{b(j')}\left((v_{k-1}-v_k)q_{j'}\right)\\
& = & (v_{b(j')}-v_i)q_{j'}+(v_i-v_{b(j')})q_{j'}\\
& = & 0,
\end{eqnarray*}
where the inequality follows from the monotonicity of $\bf X$ and the fact that $q_{j'}\leq q_{f(b(j'))}$ by the definition of $f_{\bf X}$.

Finally, for any item $j'$ with $b_{j'}=0$, for which it holds $\widetilde{p}_{j'}=\infty$, $u_{ij}\geq u_{ij'}$ trivially holds.

Thus, in order to conclude the proof, we are just left to show that
$u_i\geq 0$ for each $i\in W({\bf X})$. To this aim, note that, for
each $j'\in X_{\beta({\bf X})}$, it holds $u_{\beta({\bf X})j'}=0$
by definition of $\widetilde{{\bf p}}$, which yields $u_{\beta({\bf
X})}=0$. Let $j'$ be any item belonging to $X_{\beta({\bf X})}$.
Since, as we have shown, for each buyer $i\in W({\bf X})$ and item
$j\in X_i$, it holds $u_{ij}\geq u_{ij'}$, it follows that
$u_i=\sum_{j\in X_i}(v_i q_j -\widetilde{p}_j)\geq d_i(v_i q_{j'}
-\widetilde{p}_{j'})\geq d_i(v_{\beta({\bf X})} q_{j'}
-\widetilde{p}_{j'})=d_i u_{\beta({\bf X})j'}=0$ and this concludes
the proof.\qed
\end{proof}

\section{Results for Generic Instances}
In this section, we show that it is hard to approximate the {\sf
RMPSD} to a factor $O(m^{1-\epsilon})$ for any $\epsilon>0$, even when considering related valuations,
whereas a simple $m$-approximation algorithm can be designed for unrelated valuations.

\subsection{Inapproximability Result}
For an integer $k>0$, we denote with $[k]$ the set $\{1,\ldots,k\}$.
Recall that an instance of the {\sf Partition} problem is made up of
$k$ strictly positive numbers $q_1,\ldots,q_k$ such that $\sum_{i\in
[k]}q_i=Q$, where $Q>0$ is an even number. It is well-known that
deciding whether there exists a subset $J\subset [k]$ such that
$\sum_{i\in J}q_i=Q/2$ is an $\sf NP$-complete problem. The
inapproximability result that we derive in this subsection is
obtained through a reduction from a specialization of the {\sf
Partition} problem, that we call {\sf Constrained Partition}
problem, which we define in the following.

An instance of the {\sf Constrained Partition} problem is made up of
an even number $k$ of non-negative numbers $q_1,\ldots,q_k$ such
that $\sum_{i\in [k]}q_i=Q$, where $Q$ is an even number and $\frac
3 2 \min_{i\in [k]}\{q_i\}\geq\max_{i\in [k]}\{q_i\}$. In this case,
we are asked to decide whether there exists a subset $J\subset [k]$,
with $|J|=k/2$, such that $\sum_{i\in J}q_i=Q/2$.

\begin{lemma}\label{hardnesslemma}
The {\sf Constrained Partition} problem is $\sf NP$-complete.
\end{lemma}

\begin{proof}
Let $I=\{q_1,\ldots,q_k\}$ be an instance of the {\sf Partition}
problem and denote with $q_{min}=\min_{i\in [k]}\{q_i\}$ and
$q_{max}=\max_{i\in [k]}\{q_i\}$. We construct an instance
$I'=\{q'_1,\ldots,q'_{k'}\}$ of the {\sf Constrained Partition}
problem as follows: set $k'=2k$, then, for each $i\in [k]$, set
$q'_i=q_i+2q_{max}$, while, for each $k+1\leq i\leq k'$, set
$q'_i=2q_{max}$. It is easy to see that, by construction, it holds
that $k'$ is an even number, $\frac 3 2 \min_{i\in [k']}\{q'_i\}\geq
3 q_{max}=\max_{i\in [k']}\{q'_i\}$ and that $\sum_{i\in
[k']}q'_i=\sum_{i\in [k]}q_i+2k'q_{max}=Q+2k'q_{max}$ is an even
number, so that $I'$ is a valid instance of the {\sf Constrained
Partition} problem.

In order to show the claim, we have to prove that there exists a positive answer to $I$ if and only if there exists a positive answer to $I'$.

To this aim, let $J\subset [k]$, with $\sum_{i\in J}q_i=Q/2$, be a
positive answer to $I$. Let $J'\subseteq\{k+1,\ldots,k'\}$, with
$|J'|=k-|J|$, be any set of $k-|J|$ numbers of value $2q_{max}$.
Note that, by the definition of $k'$ and the fact that $|J|<k$,
$J'\neq\emptyset$. We claim that the set $J\cup J'$ is a positive
answer to $I'$. In fact, it holds $|J\cup J'|=k$ and $\sum_{i\in
J\cup J'}q'_i=\sum_{i\in
J}(q_i+2q_{max})+2q_{max}(k-|J|)=Q/2+k'q_{max}$.

Now, let $J'\subset [k']$, with $\sum_{i\in J'}q'_i=Q/2+k'q_{max}$,
be a positive answer to $I'$. Note that, since $k'=2k$, it holds
$\sum_{i=k+1}^{k'}q'_i= k'q_{max}$. Hence, since $Q>0$, there must
exist at least one index $i\in J'$ such that $i\in [k]$. Let
$J=\{i\in J':i\in [k]\}\neq\emptyset$ be the set of all such
indexes. We claim that $J$ is a positive answer to $I$. In fact, it
holds $\sum_{i\in J}q_i=\sum_{i\in J'}q'_i-k'q_{max}=Q/2$.\qed
\end{proof}

We can now proceed to show our first inapproximability result, by
means of the following reduction. Given an integer $k\geq 3$,
consider an instance $I$ of the {\sf Constrained Partition}
problemwith $2(k-1)$ numbers $q_1,\ldots,q_{2(k-1)}$ such that
$\sum_{i=1}^{2(k-1)}q_i=Q$ and define $q_{min}=\min_{i\in
[2(k-1)]}\{q_i\}$. Remember that, by definition, $Q$ is even and it
holds $\frac 3 2 q_{min}\geq\max_{i\in [2(k-1)]}\{q_i\}$. Note that,
this last property, together with $Q\geq 2(k-1)q_{min}$, implies
that $q_j\leq\frac{3Q}{4(k-1)}<\frac Q 2$ for each $j\in [2(k-1)]$
since $k\geq 3$.

For any $\epsilon>0$, define
$\alpha=\left\lceil\frac{2}{\epsilon}\right\rceil+1$ and
$\lambda=k^\alpha$. Note that, by definition, $\lambda\geq k^2$. We
create an instance $I'$ of the {\sf RMPSD} as follows. There are
$n=5$ buyers and $m=\lambda+k-1$ items divided into four groups:
%\begin{itemize}
%\item[$\bullet$]
$k$ items of quality $Q$,
%\item[$\bullet$]
one item of quality $Q/2$,
%\item[$\bullet$]
$2(k-1)$ items of qualities $q_i$, with $i\in [2(k-1)]$, inherited from $I$,
%\item[$\bullet$]
and $\lambda-2k$ items of quality $\overline{q}:=\frac{q_{min}}{100}>0$.
%\end{itemize}
The five buyers are such that
%\begin{itemize}
%\item[$\bullet$]
$v_1=2$ and $d_1=k$,
%\item[$\bullet$]
$v_2=1+\frac 1 \lambda\frac{Q-2k\overline{q}+kQ(\lambda+1)/2}{Qk+Q-2k\overline{q}+\lambda\overline{q}}$ and $d_2=\lambda$,
%\item[$\bullet$]
$v_3=1+\frac 1 \lambda$ and $d_3=k$,
%\item[$\bullet$]
$v_4=1+\frac 1 \lambda\frac{Q-k\overline{q}}{Q+(\lambda-2k)\overline{q}}$ and $d_4=\lambda-k$,
%\item[$\bullet$]
$v_5=1$ and $d_5=\lambda-2k$.
%\end{itemize}

Note that it holds $v_i>v_{i+1}$ for each $i\in [4]$. In fact,
$v_4>1=v_5$, since $\lambda>2k$ and $Q\geq
2(k-1)q_{min}=200(k-1)\overline{q}>k\overline{q}$ for $k\geq 2$.
Moreover, $v_4<1+\frac 1 \lambda$, since $\lambda>k$ implies
$Q-k\overline{q}<Q+(\lambda-2k)\overline{q}$. Finally, $v_2>1+\frac
1 \lambda$, since
$\lambda>2=\frac{kQ}{k(Q-Q/2)}>\frac{kQ}{kQ-2\overline{q}}$ implies
$Q-2k\overline{q}+\frac{kQ(\lambda+1)}{2}>Qk+Q-2k\overline{q}+\lambda\overline{q}$
and $v_2<2=v_1$, since $\lambda>\frac k 2 +1$ implies
$Q-2k\overline{q}+\frac{kQ(\lambda+1)}{2}<\lambda(Qk+Q-2k\overline{q}+\lambda\overline{q})$.

Our aim is to show that, if there exists a positive answer to $I$,
then there exists an envy-free outcome for $I'$ of revenue at least
$(\lambda-2k)\overline{q}$, while, if a positive answer to $I$ does
not exists, then no envy-free outcome of revenue greater than
$6(k+3)(k-1)q_{min}$ can exist for $I'$.

\begin{lemma}\label{lemma1}
If there exists a positive answer to $I$, then there exists an envy-free outcome for $I'$ of revenue greater than $(\lambda-2k)\overline{q}$.
\end{lemma}

\begin{proof}
Consider the allocation vector $\bf X$ such that $X_1$ is made up of
$k$ items of quality $Q$, $X_3$ contains the item of quality $Q/2$
plus the $k-1$ items forming a positive answer to $I$, $X_5$ is made
up of the $\lambda-2k$ items of quality $\overline{q}$ and
$X_2=X_4=\emptyset$. Note that $\bf X$ is monotone. We show that the
outcome $({\bf X},\widetilde{{\bf p}})$ is envy-free.

According to the price vector $\widetilde{\bf p}$, it holds
$\widetilde{p}_j=\frac{(3\lambda+1)Q-2\overline{q}}{2\lambda}$ for
each $j\in X_1$,
$\widetilde{p}_j=\frac{(\lambda+1)q_j-\overline{q}}{\lambda}$ for
each $j\in X_3$ and $\widetilde{p}_j=\overline{q}$ for each $j\in
X_5$.

Because of Lemma \ref{prezzatura}, in order to show that $({\bf
X},\widetilde{{\bf p}})$ is envy-free, we only need to prove that,
for each buyer $i\notin W({\bf X})$ and $T\subseteq M$ with
$|T|=d_i$, it holds $\sum_{j\in T}u_{ij}\leq 0$. Note that the
buyers not belonging to $W({\bf X})$ are buyers $2$ and $4$.

For buyer $2$, since there are exactly $\lambda$ items having a
non-infinite price, it follows that $T=X_1\cup X_3\cup X_5$ is the
only set of items of cardinality $d_2$ which can give buyer $2$ a
non-negative utility. It holds
\begin{displaymath}
\begin{array}{ll}
& \displaystyle\sum_{j\in T}\left(v_2 q_j-\widetilde{p}_j\right)\\
= & \left(1+\frac 1 \lambda\frac{Q-2k\overline{q}+\frac{kQ}{2}(\lambda+1)}{Qk+Q-2k\overline{q}+\lambda\overline{q}}\right)(kQ+Q+(\lambda-2k)\overline{q})\\
& -\frac{k((3\lambda+1)Q-2\overline{q})}{2\lambda}-\frac{(\lambda+1)Q-k\overline{q}}{\lambda}-(\lambda-2k)\overline{q}\\
= & 0.
\end{array}
\end{displaymath}

For buyer $4$, for each pair of items $(j,j')$ with $j\in X_1$ and
$j'\in X_3$, it holds $u_{4j}<u_{4j'}$, while, for each pair of
items $(j',j'')$ with $j'\in X_3$ and $j''\in X_5$, it holds
$u_{4j'}<u_{4j''}$. In fact, we have
\begin{eqnarray*}
u_{4j'}-u_{4j} & = & v_4 q_j - v_4 Q -q_j\left(1+\frac 1 \lambda\right)+\frac Q 2\left(3+\frac 1 \lambda\right)\\
& > & \frac 1 \lambda\left(\frac Q 2 -q_j\right)\\
& \geq & 0,
\end{eqnarray*}
where the first inequality follows from $1<v_4<3/2$ and the second one follows from $q_j\leq Q/2$ for each $j\in X_3$; and
\begin{eqnarray*}
u_{4j''}-u_{4j'} & = & v_4 \overline{q} -\overline{q} - v_4 q_j +q_j+\frac{q_j}{\lambda}-\frac{\overline{q}}{\lambda}\\
& = & (q_j-\overline{q})\left(1+\frac 1 \lambda -v_4\right)\\
& > & 0,
\end{eqnarray*}
where the inequality follows from $v_4<1+1/\lambda$ and $q_j>\overline{q}$ for each $j\in X_3$.

Hence, the set of items of cardinality $d_4$ which gives the highest utility to buyer $4$ is $T=X_3\cup X_5$. It holds
\begin{displaymath}
\begin{array}{ll}
& \displaystyle\sum_{j\in T}\left(v_4 q_j-\widetilde{p}_j\right)\\
= & \left(1+\frac 1 \lambda\frac{Q-k\overline{q}}{Q+\overline{q}(\lambda-2k)}\right)(Q+(\lambda-2k)\overline{q})-\frac{(\lambda+1)Q-k\overline{q}}{\lambda}-(\lambda-2k)\overline{q}\\
= & 0.
\end{array}
\end{displaymath}
Thus, we can conclude that the outcome $({\bf X},\widetilde{{\bf p}})$ is envy-free and it holds $rev({\bf X},\widetilde{{\bf p}})>(\lambda-2k)\overline{q}$.\qed
\end{proof}

Now we stress the fact that, in any envy-free outcome $({\bf X},{\bf
p})$ for $I'$ such that $rev({\bf X},{\bf p})>0$, it must be
$X_1\neq\emptyset$. In fact, assume that there exists an envy-free
outcome $({\bf X},{\bf p})$ such that $X_1=\emptyset$ and
$X_i\neq\emptyset$ for some $2\leq i\leq 5$, then, since $d_1\leq
d_i$ and $v_1>v_i$ for each $2\leq i\leq 5$, it follows that there
exists a subset of $d_1$ items $T$ such that $u_1>u_i\geq 0$, which
contradicts the envy-freeness of $({\bf X},{\bf p})$. As a
consequence of this fact and of the definition of the demand vector,
it follows that each possible envy-free outcome $({\bf X},{\bf p})$
for $I'$ can only fall into one of the following three cases:
\begin{enumerate}
\item $X_1\neq\emptyset$ and $X_i=\emptyset$ for each $2\leq i\leq 5$,
\item $X_1,X_3\neq\emptyset$ and $X_2,X_4,X_5=\emptyset$,
\item $X_1,X_3,X_5\neq\emptyset$ and $X_2,X_4=\emptyset$.
\end{enumerate}
Note that, for each envy-free outcome $({\bf X},{\bf p})$ falling
into one of the first two cases, it holds $rev({\bf X},{\bf p})\leq
v_1kQ+v_3\frac 3 2 Q\leq Q(2k+3)\leq (2k+3)2(k-1)\frac 3 2
q_{min}=6(k+3)(k-1)q_{min}$. In the remaining of this proof, we will
focus only on outcomes falling into case $(3)$.

First of all, we show that, if any such an outcome is envy-free, then the sum of the qualities of the items assigned to buyer $3$ cannot exceed $Q$.

\begin{lemma}\label{lemma2}
In any envy-free outcome $({\bf X},{\bf p})$ falling into case $(3)$, it holds $\sum_{j\in X_3}q_j\leq Q$.
\end{lemma}

\begin{proof}
Let $({\bf X},{\bf p})$ be an envy-free outcome falling into case
$(3)$ and assume, for the sake of contradiction, that $\sum_{j\in
X_3}>Q$. Note that, in this case, because of Lemma
\ref{orderedqualities} and the fact that no subset of $k$ items
inherited from $I$ can sum a total quality greater than $Q$, $X_3$
must contain the item of quality $Q/2$ and $X_1$ must contain all
items of quality $Q$.

By the feasibility of $({\bf X},{\bf p})$, it holds $u_5\geq 0$
which implies that there exists one item $j'\in X_5$ such that
$p_{j'}\leq q_{j'}$. Moreover, by the envy-freeness of $({\bf
X},{\bf p})$, for each $j\in X_3$, it holds
$u_{3j}=\frac{\lambda+1}{\lambda}q_j-p_j\geq
u_{3j'}=\frac{\lambda+1}{\lambda}q_{j'}-p_{j'}\geq\frac{\lambda+1}{\lambda}q_{j'}-q_{j'}=\frac{q_{j'}}{\lambda}$
which implies
$p_j\leq\frac{\lambda+1}{\lambda}q_j-\frac{q_{j'}}{\lambda}\leq\frac{\lambda+1}{\lambda}q_j-\frac{\overline{q}}{\lambda}$
for each $j\in X_3$. Let $j''$ denote the item of quality $Q/2$.
Since $j''\in X_3$, it follows that
$p_{j''}\leq\frac{\lambda+1}{\lambda}\frac Q
2-\frac{\overline{q}}{\lambda}$. Again, by the envy-freeness of
$({\bf X},{\bf p})$, for each $j\in X_1$, it holds
$u_{1j}=2Q-p_j\geq u_{1j''}=Q-p_{j''}\geq
Q-\frac{\lambda+1}{\lambda}\frac Q 2+\frac{\overline{q}}{\lambda}$
which implies $p_j\leq\frac{3Q\lambda+Q-2\overline{q}}{2\lambda}$.

Define $T=X_1\cup X_3\cup X_5$ and let us compute the utility that buyer $2$ achieves if she is assigned set $T$ such that $|T|=\lambda=d_2$. It holds
\begin{displaymath}
\begin{array}{lcl}
u_2 & = & \displaystyle\sum_{j\in T}(v_2 q_j-p_j)\\
& = & v_2\displaystyle\sum_{j\in X_5}q_j-\displaystyle\sum_{j\in X_5}p_j+v_2\displaystyle\sum_{j\in X_3}q_j-\displaystyle\sum_{j\in X_3}p_j
+v_2\displaystyle\sum_{j\in X_1}q_j-\displaystyle\sum_{j\in X_1}p_j\\
& \geq & \left(\frac 1 \lambda\frac{Q-2k\overline{q}+\frac{kQ}{2}(\lambda+1)}{Q+(\lambda-2k)\overline{q}}\right){\displaystyle\sum_{j\in X_5}q_j}+(v_2-v_3){\displaystyle\sum_{j\in X_3}q_j}+\frac{k\overline{q}}{\lambda}+k\left(v_2 Q-\frac{3Q\lambda+Q-2\overline{q}}{2\lambda}\right)\\
& > & \frac{(\lambda-2k)(Q-2k\overline{q}+\frac{kQ}{2}(\lambda+1))}{\lambda(Q+(\lambda-2k)\overline{q})}+\frac{(Qk(\lambda-1)-2\lambda\overline{q})Q}{2\lambda(Q(k+1)+(\lambda-2k)\overline{q})}+\frac{k\overline{q}}{\lambda}+k\left(v_2 Q-\frac{3Q\lambda+Q-2\overline{q}}{2\lambda}\right)\\
& = & 0,
\end{array}
\end{displaymath}
where the first inequality comes from the fact that, for each $j\in
X_1$, it holds $q_j=Q$ and
$p_j\leq\frac{3Q\lambda+Q-2\overline{q}}{2\lambda}$, the fact that
$u_5\geq 0$ implies $\sum_{j\in X_5}q_j\geq\sum_{j\in X_5}p_j$ and
the fact that
$p_j<\frac{\lambda+1}{\lambda}q_j-\frac{\overline{q}}{\lambda}$ for
each $j\in X_3$, while the second inequality comes from the fact
that $\sum_{j\in X_5}q_j\geq(\lambda-2k)\overline{q}$ and
$\sum_{j\in X_3}q_j>Q$.

Hence, since there exists a subset of $d_2$ items for which buyer
$2$ gets a strictly positive utility and buyer $2$ is not a winner
in $\bf X$, it follows that the outcome $({\bf X},{\bf p})$ cannot
be envy-free, a contradiction.\qed
\end{proof}

On the other hand, we also show that, for any envy-free outcome
$({\bf X},{\bf p})$ falling into case $(3)$, the sum of the
qualities of the items assigned to buyer $3$ cannot be smaller than
$Q$.

\begin{lemma}\label{lemma3}
In any envy-free outcome $({\bf X},{\bf p})$ falling into case $(3)$, it holds $\sum_{j\in X_3}q_j\geq Q$.
\end{lemma}

\begin{proof}
Let $({\bf X},{\bf p})$ be an envy-free outcome falling into case $(3)$ and assume, for the sake of contradiction, that $\sum_{j\in X_3}<Q$.

By the feasibility of $({\bf X},{\bf p})$, it holds $u_5\geq 0$
which implies that there exists one item $j'\in X_5$ such that
$p_{j'}\leq q_{j'}$. Moreover, by the envy-freeness of $({\bf
X},{\bf p})$, for each $j\in X_3$, it holds
$u_{3j}=\frac{\lambda+1}{\lambda}q_j-p_j\geq
u_{3j'}=\frac{\lambda+1}{\lambda}q_{j'}-p_{j'}\geq\frac{\lambda+1}{\lambda}q_{j'}-q_{j'}=\frac{q_{j'}}{\lambda}$
which implies
$p_j\leq\frac{\lambda+1}{\lambda}q_j-\frac{q_{j'}}{\lambda}\leq\frac{\lambda+1}{\lambda}q_j-\frac{\overline{q}}{\lambda}$
for each $j\in X_3$.

Define $T=X_3\cup X_5$ and let us compute the utility that buyer $4$ achieves if she is assigned set $T$ such that $|T|=\lambda-k=d_4$. It holds
\begin{displaymath}
\begin{array}{lcl}
u_4 & = & \displaystyle\sum_{j\in T}(v_4 q_j-p_j)\\
& = & v_4\displaystyle\sum_{j\in X_5}q_j-\displaystyle\sum_{j\in X_5}p_j+v_4\displaystyle\sum_{j\in X_3}q_j-\displaystyle\sum_{j\in X_3}p_j\\
& \geq & \left(\frac 1 \lambda\frac{Q-k\overline{q}}{Q+(\lambda-2k)\overline{q}}\right){\displaystyle\sum_{j\in X_5}q_j}+(v_4-v_3){\displaystyle\sum_{j\in X_3}q_j}+\frac{k\overline{q}}{\lambda}\\
& > & \frac{(\lambda-2k)(Q-k\overline{q})\overline{q}}{\lambda(Q+(\lambda-2k))\overline{q}}
-\frac{Q(\lambda-k)\overline{q}}{\lambda(Q+(\lambda-2k)\overline{q})}+\frac{k\overline{q}}{\lambda}\\
& = & 0
\end{array}
\end{displaymath}
where the first inequality comes from the fact that $u_5\geq 0$
implies $\sum_{j\in X_5}q_j\geq\sum_{j\in X_5}p_j$ and the fact that
$p_j\leq\frac{\lambda+1}{\lambda}q_j-\frac{\overline{q}}{\lambda}$
for each $j\in X_3$, while the second inequality comes from the fact
that $\sum_{j\in X_5}q_j\geq (\lambda-2k)\overline{q}$ and
$\sum_{j\in X_3}q_j<Q$.

Hence, since there exists a subset of $d_4$ items for which buyer
$4$ gets a strictly positive utility and buyer $4$ is not a winner
in $\bf X$, it follows that the outcome $({\bf X},{\bf p})$ cannot
be envy-free, a contradiction.\qed
\end{proof}

As a consequence of Lemmas \ref{lemma2} and \ref{lemma3}, it follows
that there exists an envy-free outcome $({\bf X},{\bf p})$ falling
into case $(3)$ only if $\sum_{j\in X_3}q_j=Q$. Since, as we have
already observed, in such a case the item of quality $Q/2$ has to
belong to $X_3$, it follows that there exists an envy-free outcome
$({\bf X},{\bf p})$ falling into case $(3)$ only if there are $k-1$
items inherited from $I$ whose sum is exactly $Q/2$, that is, only
if $I$ admits a positive solution.

Any envy-free outcome not falling into case $(3)$ can raise a
revenue of at most $6(k+3)(k-1)q_{min}$. Hence, if there exists a
positive answer to $I$, then, by Lemma~\ref{lemma1}, there exists a
solution to $I'$ of revenue greater than $(\lambda-2k)\overline{q}$,
while, if there is no positive answer to $I$, then there exists no
solution to $I'$ of revenue more than $6(k+3)(k-1)q_{min}$.

Thus, if there exists an $r$-approximation algorithm for the {\sf
RMPSD} with $r\leq
\frac{(\lambda-2k)q_{min}}{600(k+3)(k-1)q_{min}}$, it is then
possible to decide in polynomial time the {\sf Constrained
Partition} problem, thus implying {\sf P} = {\sf NP}. Since, by the
definition of $\alpha$,
$\frac{\lambda-2k}{600(k+3)(k-1)}=O\left(k^{\alpha-2}\right)=O\left(m^{1-2/\alpha}\right)$
and $m^{1-\epsilon}<m^{1-2/\alpha}$, the following theorem holds.

\begin{theorem}\label{maininapprox}
For any $\epsilon>0$, the {\sf RMPSD} cannot be approximated to a factor $O(m^{1-\epsilon})$ unless ${\sf P}={\sf NP}$.
\end{theorem}

We stress that this inapproximability result heavily relies on the presence
of two useless buyers, namely buyers $2$ and $4$, who cannot be
winners in any envy-free solution. This situation suggests that better approximation guarantees may be possible for proper instances, as we will show in the next section.

\subsection{The Approximation Algorithm}
In this subsection, we design a simple $m$-approximation algorithm
for the generalization of the {\sf RMPSD} in which the buyers have
unrelated valuations. The inapproximability result given in Theorem
\ref{maininapprox} shows that, asymptotically speaking, this is the
best approximation one can hope for unless {\sf P} $=$ {\sf NP}.

For each $i\in N$, let $T_i=\textrm{argmax}_{T\subseteq
M:|T|=d_i}\left\{\sum_{j\in T}v_{ij}\right\}$ be the set of the
$d_i$ best items for buyer $i$ and define $R_i=\left(\sum_{j\in
T_i}v_{ij}\right)/d_i$. Let $i^*$ be the index of the buyer with the
highest value $R_i$. Consider the algorithm {\sf best} which returns
the outcome $(\overline{{\bf X}},\overline{\bf p})$ such that
$\overline{X}_{i^*}=T_{i^*}$, $\overline{X}_i=\emptyset$ for each
$i\neq i^*$, $\overline{p}_j=R_{i^*}$ for each $j\in T_{i^*}$ and
$\overline{p}_j=\infty$ for each $j\notin T_{i^*}$. It is easy to see that the computational complexity of Algorithm {\sf best} is $O(nm)$.

\begin{theorem}\label{apxgeneral}
Algorithm {\sf best} returns an $m$-approximate solution for the {\sf RMPSD} with unrelated valuations.
\end{theorem}

\begin{proof}
It is easy to see that the outcome $(\overline{{\bf
X}},\overline{\bf p})$ returned by Algorithm {\sf best} is
feasible.In order to prove that it is also envy-free, we just need
to show that, for each buyer $i\neq i^*$ with $d_i\leq d_{i^*}$ and
each $T_i\subseteq T_{i^*}$ of cardinality $d_i$, it holds
$\sum_{j\in T_i}(v_{ij}-p_j)\leq 0$. Assume, for the sake of
contradiction, that there exists a set $T_i$ of cardinality $d_i$
such that $\sum_{j\in T_i}(v_{ij}-\overline{p}_j)>0$.

We obtain $0<\sum_{j\in T_i}(v_{ij}-\overline{p}_j)=\sum_{j\in
T_i}v_{ij}-d_i R_{i^*}\leq d_i R_i-d_i R_{i^*}=d_i(R_i-R_{i^*})$
which implies $R_i>R_{i^*}$, a contradiction. Hence,
$(\overline{{\bf X}},\overline{\bf p})$ is envy-free.

As to the approximation guarantee, note that $rev(\overline{{\bf
X}},\overline{\bf p})=d_{i^*}R_{i^*}\geq R_{i^*}$. The maximum
possible revenue achievable by any outcome $({\bf X},{\bf p})$, not
even an envy-free one, is at most $\sum_{i\in N}\sum_{j\in
X_i}v_{ij}\leq\sum_{i\in W({\bf X})}(d_i R_i)\leq m R_{i^*}$, which
yields the claim.\qed
\end{proof}

\section{Results for Proper Instances}

Given a proper instance $I=({\bf V},{\bf D},{\bf Q})$, denote with
$\delta$ the number of different values in $\bf V$ and, for
each $k\in [\delta]$, let
$A_k\subseteq N$ denote the set of buyers with the $k$th highest value and $v(A_k)$ denote
the value of all buyers in $A_k$. For $k\in [\delta]$, define $A_{\leq
k}=\bigcup_{h=1}^k A_h$, $A_{\geq k}=\bigcup_{h=k}^\delta A_h$,
$A_{>k}=A_{\geq k}\setminus A_k$ and $A_{<k}=A_{\leq k}\setminus
A_k$, while, for each subset of buyers $A\subseteq N$, define
$d(A)=\sum_{i\in A}d_i$. Let $\delta^*\in [\delta]$ be the minimum
index such that $d(A_{\leq \delta^*})>m$ and let
$\widetilde{A}\subset A_{\delta^*}$ be a subset of buyers in
$A_{\delta^*}$ such that $$\widetilde{A}=\textrm{argmax}_{A\subset
A_{\delta^*}:d(A)+d(A_{< \delta^*})\leq m}\left\{d(A)\right\}.$$ In
other words $\widetilde{A}$ is the subset of buyers in
$A_{\delta^*}$ that feasibly extends $A_{< \delta^*}$ (i.e., such
that the sum of the requested items of buyers in $A_{< \delta^*}
\cup \widetilde{A}$ is at most $m$) and maximizes the number of
allocated items.

Note that any instance $I$ for which $\delta^*$ does not exist can
be suitably extended with a dummy buyer $n+1$, such that
$v_{n+1}<v_n$ and $d_{n+1}=m+1$, which is equivalent in the sense
that it does not change the set of envy-free outcomes of $I$. Hence,
in this section, we will always assume that $\delta^*$ is
well-defined for each proper instance of the {\sf RMPSD}.

For our purposes we need to break ties among values of the buyers in
$A_{\delta^*}$ in such a way that each buyer in $\widetilde{A}$
comes before any buyer in $A_{\delta^*}\setminus \widetilde{A}$. In
order to achieve this task, we need to explicitly compute the set of
buyers $\widetilde{A}$. Such a computation can be done by reducing
this problem to the {\sf knapsack} problem. It is easy to see that,
in this case, the well-known pseudo-polynomial time algorithm for
{\sf knapsack} is polynomial in the dimensions of $I$, as $d_i\leq
m$ for every $i\in N$.
%In fact, in order to compute $\widetilde{A}$, we need to detect, among all subsets of buyers in $A_{\delta^*}$ whose global demand does not exceed $m-d(A_{<\delta^*})$, the one of maximum global demand. Hence, if we associate to each buyer $i\in A_{\delta^*}$ an object of the {\sf knapsack} problem with profit $d_i$ and volume $d_i$, any possible set $\widetilde{A}$ coincides with any optimal solution to this {\sf knapsack} problem in which we have $k=|A_{\delta^*}|\leq n$ objects and knapsack's volume set to $W=m-d(A_{<\delta^*})$. Recall that the {\sf knapsack} problem, although {\sf NP}-hard, admits a pseudo-polynomial algorithm of complexity $O(kW)$. In such a case, since $k\leq n$ and $W\leq m$, the complexity of the {\sf knapsack} problem, and thus the complexity of computing the set $\widetilde{A}$, is polynomial in the dimensions of the given instance $I$.

Because of the above discussion, from now on we can assume that ties
among values of the buyers in $A_{\delta^*}$ are broken in such a
way that each buyer in $\widetilde{A}$ comes before any buyer in
$A_{\delta^*}\setminus \widetilde{A}$. For each $k\in [\delta^*]$,
define
\begin{displaymath}
\alpha(k)=\left\{
\begin{array}{ll}
\max\{i\in A_k\} & \textrm{  if }k\in [\delta^*-1],\\
\max\{i\in \widetilde{A}\} & \textrm{  if }k=\delta^*.
\end{array}\right.
\end{displaymath}
By the definition of $\delta^*$ and $\widetilde{A}$ and by the tie
breaking rule imposed on the buyers in $A_{\delta^*}$, it follows
that $\sum_{i=1}^{\alpha(k)}d_i\leq m$ for each $k\in [\delta^*]$.

We say that an allocation vector $\bf X$ is an $h$-prefix of $I$, with $h\in [\alpha(\delta^*)]$, if $\bf X$ is monotone and $i\in W({\bf X})$ if and only if $i\in [h]$.

\subsection{Computing an $h$-Prefix of $I$ of Maximum Revenue}
Let $\bf X$ be an $h$-prefix of $I$.
We show that $({\bf X},\widetilde{{\bf p}})$ is an envy-free outcome.

\begin{lemma}\label{prezzienvyfree}
The outcome $({\bf X},\widetilde{{\bf p}})$ is envy-free.
\end{lemma}

\begin{proof}
Since $\bf X$ is monotone, by exploiting Lemma \ref{prezzatura}, we
only need to prove that for each buyer $i\notin W({\bf X})$ and set
$T\subseteq M$ of cardinality $d_i$, it holds $\sum_{j\in
T}u_{ij}\leq 0$. Note that $i\notin W({\bf X})$ if and only if
$i>h$.

For each $i>h$, it holds $v_i\leq v_h$. Moreover, for each $j$ such
that $b(j)=h$, it holds $u_{hj}=0$. Since, because of Lemma
\ref{prezzatura}, $u_{hj}\geq u_{hj'}$ for any item $j'\in M({\bf
X})$, it follows that $u_{hj'}=v_h q_{j'}-\widetilde{p}_{j'}\leq 0$
for each $j'\in M({\bf X})$. Hence, for each $j'\in M({\bf X})$, it
holds $u_{ij'}=v_i q_{j'}-\widetilde{p}_{j'}\leq v_h
q_{j'}-\widetilde{p}_{j'}\leq 0$ and this concludes the proof.\qed
\end{proof}

Given an allocation vector $\bf X$, for each $i\in [\delta]$, denote
with $M_i({\bf X})=\{j\in M({\bf X}):v_{b(j)}=v(A_i)\}$ the set of
items allocated to the buyers with the $i$th highest value in ${\bf
V}$. Recall that, since $\bf X$ is an $h$-prefix of $I$, it holds
$\beta({\bf X})=h$. The following lemma gives a lower bound on the
revenue generated by the outcome $({\bf X},\widetilde{{\bf p}})$.

\begin{lemma}\label{revenueenvyfree}
$rev({\bf X},\widetilde{{\bf p}})\geq v_h\sum_{j\in M_{h}({\bf X})}q_j$.
\end{lemma}

\begin{proof}
By the definition of $\widetilde{\bf p}$, it follows that $rev({\bf X},\widetilde{{\bf p}})\geq\sum_{j\in M_h({\bf X})}\widetilde{p}_j=v_h\sum_{j\in M_h({\bf X})}q_j$.\qed
\end{proof}

We now prove a very important result stating that the price vector
$\widetilde{\bf p}$ is the best one can hope for when overpricing is
not allowed. Such a result, of independent interest, plays a crucial
role in the proof of the approximation guarantee of the algorithm we
define in this section.

\begin{lemma}\label{optimalwhennooverprice}
Let $\bf X$ be an $h$-prefix of $I$. Then $({\bf X},\widetilde{{\bf p}})$ is an optimal envy-free outcome when overpricing is not allowed.
\end{lemma}

\begin{proof}
It is easy to see that the price vector $\widetilde{\bf p}$ does not
overprice any item in $M({\bf X})$. For any envy-free outcome $({\bf
X},{\bf p})$, we show by backward induction that $p_j\leq
\widetilde{p}_j$ for each $j\in M({\bf X})$.

As a base case, for all $j\in M_h({\bf X})$, it holds $p_j\leq v_h q_j=\widetilde{p}_j$ because ${\bf p}$ cannot overprice any item.

For the inductive step, consider an item $j$ such that $b(j)=i<h$
and assume the claim true for each item $j'$ such that $b(j')>i$. By
the envy-freeness of $({\bf X},{\bf p})$, it holds
$u_{ij}-u_{ij'}\geq 0$ for $j'=f(i+1)$. This implies
\begin{eqnarray*}
0 & \leq & u_{ij}-u_{ij'}\\
& = & v_i q_j-p_j-v_i q_{j'}+p_{j'}\\
& \leq & v_i q_j-p_j-v_i q_{j'}+\widetilde{p}_{j'},
\end{eqnarray*}
where the last inequality comes from the inductive hypothesis. Hence, we can conclude that
\begin{eqnarray*}
p_j & \leq & v_i(q_j-q_{j'})+\widetilde{p}_{j'}\\
& = & v_i(q_j-q_{j'})+v_{b(j')}q_{j'}-\sum_{k=b(j')+1}^h((v_{k-1}-v_k)q_{f(k)})\\
& = & v_i(q_j-q_{f(i+1)})+v_{i+1}q_{f(i+1)}-\sum_{k=i+2}^h((v_{k-1}-v_k)q_{f(k)})\\
& = & v_i q_j-\sum_{k=i+1}^h((v_{k-1}-v_k)q_{f(k)})\\
& = & \widetilde{p}_j,
\end{eqnarray*}
where the second equality comes from $j'=f(i+1)$ and $b(j')=i+1$.
This completes the induction and shows the claim.\qed
\end{proof}

We design a polynomial time algorithm {\sf ComputePrefix} which,
given a proper instance $I$ and a value $h\in[\alpha(\delta^*)]$,
outputs the $h$-prefix ${\bf X}^*_h$ such that the outcome $({\bf
X}_h^*,\widetilde{{\bf p}})$ achieves the highest revenue among all
possible $h$-prefixes of $I$.

Recall that, by definition of $h$-prefixes of $I$, the set of buyers
whose demand is to be satisfied is exactly characterized. Moreover,
once fixed a set of items which exactly satisfies the demands of the
considered buyers, by the monotonicity of $h$-prefixes of $I$, we
know exactly which items must be assigned to each buyer. Hence, in
this setting, our task becomes that of determining the set of items
maximizing the value $rev({\bf X},\widetilde{{\bf p}})$.

To this aim, we first show that this problem reduces to that of determining, for each $i\in [h]$, the item $f(i)$. In fact, it holds
\begin{eqnarray*}
rev({\bf X},\widetilde{{\bf p}}) & = & \sum_{j\in M({\bf X})}\widetilde{p}_j\\
& = & \sum_{i\in [h]}\sum_{j\in X_i}\left(v_i q_j -\sum_{k=i+1}^h\left((v_{k-1}-v_k)q_{f(k)}\right)\right)\\
& = & \sum_{i\in [h]}\left(v_i\sum_{j\in X_i}q_j\right)-\sum_{i\in [h]}\left(d_i\sum_{k=i+1}^h\left((v_{k-1}-v_k)q_{f(k)}\right)\right)\\
& = & \underbrace{\sum_{i\in [h]}\left(v_i\sum_{j\in X_i}q_j\right)}_{T_1}-\underbrace{\sum_{i=2}^h\left(\left((v_{i-1}-v_i)q_{f(i)}\right)\sum_{k=1}^{i-1}d_k\right)}_{T_2}.
\end{eqnarray*}
Note that only those items $j$ such that $j=f(i)$ for some $i\in
[h]$ contribute to the term $T_2$ and that the per quality
contribution of each item to the term $T_1$ is always strictly
positive. This implies that, once fixed all the items $j$ such that
$j=f(i)$ for each $i\in [h]$, the remaining $d_i-1$ items to be
assigned to buyer $i$ in each optimal outcome are exactly the items
$j+1,\ldots,j+d_i-1$.

Because of the above discussion, we are now allowed to concentrate
only on the problem of determining the set of best-quality items
assigned to each buyer in $[h]$ in an optimal envy-free outcome. Let
us denote with $r_{ij}$ the maximum revenue which can be achieved by
an envy-free outcome in which the best-quality item of the first $i$
buyers have been chosen among the first $j$ ones. Hence, $r_{ij}$ is
defined for $0\leq i\leq h$ and $\sum_{k=1}^{i-1}d_k+1\leq j\leq
m+1-\sum_{k=i}^h d_k$ and has the following expression:
\begin{displaymath}
r_{ij}=\left\{
\begin{array}{ll}
0 & \textrm{  if }i=0,\\
t_i q_j+\displaystyle\sum_{k=j+1}^{j+d_i-1}v_i q_k & \textrm{  if }i>0\wedge j=\displaystyle\sum_{k=1}^{i-1}d_k+1,\\
\max\{r_{i-1,j-1}+t_i q_j;r_{i,j-1}\}+\displaystyle\sum_{k=j+1}^{j+d_i-1}v_i q_k & \textrm{  if }i>0 \wedge j>\displaystyle\sum_{k=1}^{i-1}d_k+1,
\end{array}\right.
\end{displaymath}
where $t_i=v_i -(v_{i-1}-v_i)\sum_{k=1}^{i-1}d_k$ is the
contribution that item $f(i)$ gives to the revenue per each unit of
quality. Clearly, by definition, $r_{h,m+1-d_h}$ gives the maximum
revenue which can be achieved by an envy-free outcome $({\bf
X},\widetilde{{\bf p}})$ such that $W({\bf X})=[h]$. Such a
quantity, as well as the allocation vector ${\bf X}_h^*$ realizing
it, can be computed by the following dynamic programming algorithm
of complexity $O(mh)$.

\

\noindent ${\sf ComputePrefix}({\sf input}\textrm{: instance }I,\textrm{ integer }h,\ {\sf output}\textrm{: allocation vector }{\bf X}^*_h)$:\\
{\bf for each} $i=0,\dots,h$ {\bf do} $r_{ij}:=0$;\\
{\bf for each} $i=1,\ldots,h$ {\bf do}\\
$|$\hspace{0.5cm}$r_{ij}:=t_i\cdot q_j$ where $j=\sum_{k=1}^{i-1}d_k+1$;\\
$|$\hspace{0.5cm}$f_i:=j$;\\
{\bf for each} $i=1,\ldots,h$ {\bf do}\\
$|$\hspace{0.5cm}{\bf for each} $j=\sum_{k=1}^{i-1}d_k+2,\ldots,m+1-\sum_{k=i}^h d_k$ {\bf do}\\
$|$\hspace{0.5cm}$|$\hspace{0.5cm}{\bf if} $r_{i,j-1}\geq r_{i-1,j-1}+t_i\cdot q_j$ {\bf then};\\
$|$\hspace{0.5cm}$|$\hspace{0.5cm}$|$\hspace{0.5cm}$r_{ij}:=r_{i,j-1}$;\\
$|$\hspace{0.5cm}$|$\hspace{0.5cm}{\bf else}\\
$|$\hspace{0.5cm}$|$\hspace{0.5cm}$|$\hspace{0.5cm}$r_{i,j}:=r_{i-1,j-1}+t_i\cdot q_j$;\\
$|$\hspace{0.5cm}$|$\hspace{0.5cm}$|$\hspace{0.5cm}$f_i:=j$;\\
{\bf for each} $i=1,\ldots,h$ {\bf do}\\
$|$\hspace{0.5cm}$X_i:=\{f_i,f_i+1,\ldots,f_i+d_i-1\}$;\\
{\bf return} ${\bf X}^*_h=(X_1,\ldots,X_h)$;

\

Let ${\cal X}(h)$ be the set of all possible $h$-prefixes of $I$. As
a consequence of the analysis carried out in this subsection, we can
claim the following result.

\begin{lemma}\label{subroutine}
For each $h\in [\alpha(\delta^*)]$, the $h$-prefix of $I$ ${\bf
X}^*_h$ such that $rev({\bf X}^*_h,\widetilde{{\bf p}})=\max_{{\bf
X}\in {\cal X}(h)}\{rev({\bf X},\widetilde{{\bf p}})\}$ can be
computed in time $O(mh)$.
\end{lemma}

\subsection{The Approximation Algorithm}
Our approximation algorithm $\sf Prefix$ for proper instances
generates a set of prefixes of $I$ for which it computes the
allocation of items yielding maximum revenue by exploiting the
algorithm {\sf ComputePrefix} as a subroutine. Then, it returns the
solution with the highest revenue among them.

\

\noindent ${\sf Prefix}({\sf input}\textrm{: instance }I,\ {\sf output}\textrm{: allocation vector }{\bf X}^*)$:\\
$opt:=\emptyset$; $value:=-1$;\\
compute $\widetilde{A}$;\\
reorder the buyers in such a way that each $i\in\widetilde{A}$ comes before any $i'\in A_{\delta^*}\setminus \widetilde{A}$;\\
{\bf for each} $h=1,\dots,\alpha(\delta^*)$ {\bf do}\\
$|$\hspace{0.5cm}${\bf X}^*_h:={\sf ComputePrefix}(I,h)$;\\
$|$\hspace{0.5cm}{\bf if} $rev({\bf X}^*_h,\widetilde{{\bf p}})>value$ {\bf then}\\
$|$\hspace{0.5cm}$|$\hspace{0.5cm}$opt:={\bf X}^*_h$; $value:=rev({\bf X}^*_h,\widetilde{{\bf p}})$;\\
{\bf for each} $k=0,\dots,\delta^*-1$ {\bf do}\\
$|$\hspace{0.5cm}{\bf for each} $i\in A_{k+1}$ {\bf do}\\
$|$\hspace{0.5cm}$|$\hspace{0.5cm}reorder the buyers in $A_{k+1}$ in such a way that $i$ is the first buyer in $A_{k+1}$;\\
$|$\hspace{0.5cm}$|$\hspace{0.5cm}{\bf if} $d(A_{\leq k})+d_i\leq m$ {\bf then} ${\bf X}^*_k:={\sf ComputePrefix}(I,|A_{\leq k}|+1)$;$\ \ \ \ \ \ \ \ \ \ \ \ \ \ \ \ (\dag)$\\
$|$\hspace{0.5cm}$|$\hspace{0.5cm}{\bf if} $rev({\bf X}^*_k,\widetilde{{\bf p}})>value$ {\bf then}\\
$|$\hspace{0.5cm}$|$\hspace{0.5cm}$|$\hspace{0.5cm}$opt:={\bf X}^*_k$; $value:=rev({\bf X}^*_k,\widetilde{{\bf p}})$;\\
{\bf return} $opt$;

\

It is easy to see that the computational complexity of Algorithm $\sf Prefix$ is $O(n^3m)$. As a major positive contribution of this work, we show that it approximates the {\sf RMPSD} to a factor $2$ on proper instance.

\begin{theorem}\label{apxcontinuous}
The approximation ratio of Algorithm $\sf Prefix$ is $2$ when applied to proper instances.
\end{theorem}

\begin{proof}
Let $I$ be a proper instance and let $({\bf X},{\bf p})$ be its
optimal envy-free outcome. We denote with $rev({\sf Prefix})$ the
revenue of the outcome returned by Algorithm $\sf Prefix$. The proof
is divided into two cases:

\

{\noindent\bf Case (1):} ${\bf X}$ is an $h$-prefix of $I$ for some
$h\in [\alpha(\delta^*)]$.

\noindent Since ${\bf X}$ is an $h$-prefix of $I$, the outcome
$({\bf X}^*_h,\widetilde{{\bf p}})$ has to be considered by
algorithm $\sf Prefix$ as a candidate solution. It follows that
$rev({\sf Prefix})\geq rev({\bf X}^*_h,\widetilde{{\bf p}})\geq
v_h\sum_{j\in M_h({\bf X}^*_h)}q_j$ by the definition of algorithm
$\sf Prefix$ and by Lemma \ref{revenueenvyfree}.

Now, if $\sum_{j\in M_h({{\bf X}})}{p}_j\geq\frac 1 2 rev({{\bf
X}},{{\bf p}})$, the claim directly follows since, by the
feasibility of $({{\bf X}},{{\bf p}})$, it holds $\sum_{j\in
M_h({{\bf X}})}{p}_j\leq v_h\sum_{j\in M_h({{\bf X}})}q_j\leq
rev({\sf Prefix})$. Hence, assume that $\sum_{j\in M_h({{\bf
X}})}{p}_j<\frac 1 2 rev({{\bf X}},{{\bf p}})$.

Define $i'=\max\{i\in N:v_i>v_h\}$ (note that $i'$ is well-defined
because of the assumption) and ${\bf X}'$ as the $i'$-prefix of $I$
such that $X'_i={X}_i$ for each $i\in [i']$. By
Lemma~\ref{nooverprice}, it follows that $({\bf X'},{{\bf p}})$ is
an outcome without overpricing. Because of our assumption, it holds
$rev({\bf X'},{{\bf p}})>\frac 1 2 rev({{\bf X}},{{\bf p}})$ and, by
Lemma~\ref{optimalwhennooverprice}, it also holds $rev({\bf
X'},\widetilde{{\bf p}})\geq rev({\bf X'},{{\bf p}})$. Moreover,
since ${\bf X'}$ is an $i'$-prefix of $I$, by the definition of
algorithm $\sf Prefix$ and by Lemma~\ref{subroutine}, it holds
$rev({\sf Prefix})\geq rev({\bf X}^*_{i'},\widetilde{{\bf p}})\geq
rev({\bf X'},\widetilde{{\bf p}})$ which yields the claim.

\

{\noindent\bf Case (2):} ${\bf X}$ is not an $h$-prefix of $I$ for
any $h\in [\alpha(\delta^*)]$.

\noindent Let $i^*=\min\{i\in N:i\notin W({\bf X})\}$. Since ${\bf
X}$ is not an $h$-prefix of $I$ for any $h\in [\alpha(\delta^*)]$,
it follows that $\beta({{\bf X}})>i^*$.

Assume that $\sum_{i=1}^{i^*-1}\sum_{j\in {X}_i}{p}_j\geq\frac 1 2
rev({{\bf X}},{{\bf p}})$ and define ${\bf X}'$ as the
$(i^*-1)$-prefix of $I$ such that $X'_i={X}_i$ for each $i\in
[i^*-1]$ (note that our assumption implies that $(i^*-1)$-prefixes
of $I$ do exist). By Lemma~\ref{nooverprice}, it follows that $({\bf
X'},{{\bf p}})$ is an outcome without overpricing. Because of our
assumption, it holds $rev({\bf X'},{{\bf
p}})=\sum_{i=1}^{i^*-1}\sum_{j\in {X}_i}{p}_j\geq\frac 1 2 rev({{\bf
X}},{{\bf p}})$ and, by Lemma~\ref{optimalwhennooverprice}, it also
holds $rev({\bf X'},\widetilde{{\bf p}})\geq rev({\bf X'},{{\bf
p}})$. Moreover, since ${\bf X'}$ is an $(i^*-1)$-prefix of $I$, by
the definition of algorithm $\sf Prefix$ and by
Lemma~\ref{subroutine}, it holds $rev({\sf Prefix})\geq rev({\bf
X}^*_{i^*-1},\widetilde{{\bf p}})\geq rev({\bf X'},\widetilde{{\bf
p}})$ which yields the claim.

Hence, from now on, we assume that $\sum_{i=1}^{i^*-1}\sum_{j\in
{X}_i}{p}_j<\frac 1 2 rev({{\bf X}},{{\bf p}})$.

If there does not exist an $i^*$-prefix of $I$, then,
$\sum_{i>i^*:i\in W({\bf X})}d_i<d_{i^*}$. Assume that there exists
a buyer $i'\in W({\bf X})$ such that $i'<i^*$ and $d_{i'}>d_{i^*}$.
Clearly, $i'$-prefixes of $I$ do exist. Define ${\bf X}'$ as the
$i'$-prefix of $I$ such that $X'_i={X}_i$ for each $i\in [i']$. By
the definition of algorithm $\sf Prefix$ and by
Lemmas~\ref{subroutine} and \ref{revenueenvyfree}, it holds
$rev({\sf Prefix})\geq rev({\bf X}^*_{i'},\widetilde{{\bf p}})\geq
rev({\bf X'},\widetilde{{\bf p}})\geq v_{i'}\sum_{j\in X'_{i'}}q_j$.
On the other hand, it holds $\sum_{i>i^*}\sum_{j\in
X_i}q_j<d_{i^*}q_{max}$, where
$q_{max}=\max\left\{q_j:j\in\bigcup_{i>i^*}X_i\right\}$. Moreover,
$d_{i^*}q_{max}<d_{i'}\sum_{j\in X'_{i'}}q_j$ since $d_{i'}>d_{i^*}$
and $\bf X$ is monotone. Hence, we have
\begin{eqnarray*}
\frac 1 2 rev({\bf X},{\bf p}) & < & \sum_{i>i^*:i\in W({\bf X})}\sum_{j\in X_i}p_j\\
& \leq & \sum_{i>i^*:i\in W({\bf X})}\left(v_i\sum_{j\in X_i}q_j\right)\\
& \leq & \sum_{i>i^*:i\in W({\bf X})}\left(v_{i^*}\sum_{j\in X_i}q_j\right)\\
& < & v_{i^*}d_{i^*}q_{max}\\
& < & v_{i'}d_{i'}q_{max}\\
& \leq & v_{i'}\sum_{j\in X_{i'}}q_j,
\end{eqnarray*}
which yields the claim.

Assume that there does not exist any buyer $i'\in W({\bf X})$ such
that $i'<i^*$ and $d_{i'}>d_{i^*}$. Let $k$ be the index such that
$i^*\in A_k$. In this case, by the definition of proper instances,
it holds that the allocation vector ${\bf X}'$ which allocates the
best-quality items to the buyers in $A_{<k}$ and to $i^*$ is an
$h$-prefix of $I$ considered by Algorithm {\sf Prefix} at line
$(\dag)$ for which it holds $\sum_{j\in
X'_{i^*}}q_j\geq\sum_{i>i^*:i\in W({\bf X})}\sum_{j\in X_i}q_j$.
Hence, we have
\begin{eqnarray*}
\frac 1 2 rev({\bf X},{\bf p}) & < & \sum_{i>i^*:i\in W({\bf X})}\sum_{j\in X_i}p_j\\
& \leq & \sum_{i>i^*:i\in W({\bf X})}\left(v_i\sum_{j\in X_i}q_j\right)\\
& \leq & \sum_{i>i^*:i\in W({\bf X})}\left(v_{i^*}\sum_{j\in X_i}q_j\right)\\
& < & v_{i^*}d_{i^*}q_{max}\\
& \leq & v_{i^*}\sum_{j\in X'_{i^*}}q_j,
\end{eqnarray*}
which yields the claim.

If $i^*$-prefixes of $I$ do exist, define $H=\{i\in W({\bf
X}):v_i=v_{i^*}\}$ and let $i'=\min\{i:i\in H\}$ if
$H\neq\emptyset$, otherwise set $i'=i^*$. Moreover, define
$i''=\min\{i\in W({\bf X}):i>i^*\}$.

If $v_{i^*}>v_{\beta({{\bf X}})}$, then, by Lemma \ref{noholes}, it
holds $d_{i^*}>\sum_{i>i^*:i\in W({{\bf X}})} d_i$. Define ${\bf
X}'$ as the $i^*$-prefix of $I$ such that
$X'_i=\left\{1+\sum_{j=1}^{i-1}d_j,\ldots,d_i+\sum_{j=1}^{i-1}d_j\right\}$
for each $i\in [i^*]$, i.e., ${\bf X}'$ assigns the best-quality
items to the first $i^*$ buyers. Note that the set of buyers
$[i^*-1]$ belongs to $W({\bf X}')\cap W({{\bf X}})$. Moreover, since
$({{\bf X}},{{\bf p}})$ is envy-free, by Lemma
\ref{orderedqualities} and the fact that ${\bf X}'$ assigns the
first $g:=\sum_{i=1}^{i^*-1}d_i$ best-quality items to the first
$i^*-1$ buyers, it follows that $\sum_{j=g+1}^m
q_j\geq\sum_{j=f(i'')}^m q_j$. This inequality, together with
$d_{i^*}>\sum_{i>i^*:i\in W({{\bf X}})} d_i$, implies that $\sum_{i>
i^*:i\in W({{\bf X}})}\sum_{j\in {X}_i}q_j\leq\sum_{j\in
X'_{i^*}}q_j$.

Hence,we have that
\begin{eqnarray*}
\frac 1 2 rev({{\bf X}},{{\bf p}}) & < & \sum_{i>i^*:i\in W({{\bf X}})}\sum_{j\in {{X}}_i}{p}_j\\
& \leq & \sum_{i>i^*:i\in W({{\bf X}})}\left(v_i\sum_{j\in {{X}}_i} q_j\right)\\
& < & \sum_{i>i^*:i\in W({{\bf X}})}\left(v_{i^*}\sum_{j\in {{X}}_i} q_j\right)\\
& \leq & v_{i^*}\sum_{j\in X'_{i^*}}q_j\\
& \leq & v_{i^*}\sum_{j\in M_{i^*}({\bf X}')}q_j\\
& \leq & rev({\bf X}',\widetilde{{\bf p}})\\
& \leq & rev({\bf X}^*_{i^*},\widetilde{{\bf p}})\\
& \leq & rev({\sf Prefix}),
\end{eqnarray*}
which yields the claim.

If $v_{i^*}=v_{\beta({{\bf X}})}$, with $i^*\in A_k$ for some $k\in
[\delta^*]$, define ${\bf X}'$ as the $\alpha(k)$-prefix of $I$ such
that
$X'_i=\left\{1+\sum_{j=1}^{i-1}d_j,\ldots,d_i+\sum_{j=1}^{i-1}d_j\right\}$
for each $i\in [\alpha(k)]$. Note that the set of buyers $[i'-1]$
belongs to $W({\bf X}')\cap W({{\bf X}})$. Moreover, since $({{\bf
X}},{{\bf p}})$ is envy-free, by Lemma \ref{orderedqualities} and
the fact that ${\bf X}'$ assigns the first
$g':=\sum_{i=1}^{i'-1}d_i$ best-quality items to the first $i'-1$
buyers, it follows that $\sum_{j=g'+1}^m q_j\geq\sum_{j=f(i')}^m
q_j$. This inequality, together with the fact that
$\sum_{i=i'}^{\alpha(k)}d_i\geq\sum_{A\subseteq
A_k:d(A_{<k})+d(A)\leq m}d(A)$ for each $k\in [\delta^*]$ which
comes from the definition of $\delta^*$ and $\widetilde{A}$, implies
that $\sum_{i\geq i':i\in W({{\bf X}})}\sum_{j\in
{X}_i}q_j\leq\sum_{i=i'}^{\alpha(k)}\sum_{j\in X'_i}q_j$.

Hence,we have that
\begin{eqnarray*}
\frac 1 2 rev({{\bf X}},{{\bf p}}) & < & \sum_{i\geq i':i\in W({{\bf X}})}\sum_{j\in {{X}}_i}{p}_j\\
& \leq & \sum_{i\geq i':i\in W({{\bf X}})}\left(v_k\sum_{j\in {{X}}_i} q_j\right)\\
& = & \sum_{i\geq i':i\in W({{\bf X}})}\left(v_{i^*}\sum_{j\in {{X}}_i} q_j\right)\\
& \leq & v_{i^*}\sum_{i=i'}^{\alpha(k)}\sum_{j\in X'_i}q_j\\
& = & v_{i^*}\sum_{j\in M_{i^*}({\bf X}')}q_j\\
& \leq & rev({\bf X}',\widetilde{{\bf p}})\\
& \leq & rev({\bf X}^*_{i^*},\widetilde{{\bf p}})\\
& \leq & rev({\sf Prefix}),
\end{eqnarray*}
which yields the claim. %\qed
\end{proof}

We conclude this section by showing that the approximation ratio achieved by Algorithm $\sf Prefix$ is the best possible one for proper instances.

\begin{theorem}\label{last_inapprox_res}
For any $0<\epsilon\leq 1$, the {\sf RMPSD} on proper instances cannot be approximated to a factor $2-\epsilon$ unless {\sf P} $=$ {\sf NP}.
\end{theorem}

\begin{proof}
For an integer $k\geq 3$, consider an instance $I$ of the {\sf Constrained Partition} problem with $2(k-1)$ numbers $q_1,\ldots,q_{2(k-1)}$ such that $\sum_{i=1}^{2(k-1)}q_i=Q$ and define $q_{min}=\min_{i\in [2(k-1)]}\{q_i\}$ and $q_{max}=\max_{i\in [2(k-1)]}\{q_i\}$. Remember that, by definition, $Q$ is even and it holds $\frac 3 2 q_{min}\geq q_{max}$. Also in this case, as observed in the proof of Theorem \ref{maininapprox}, it holds $q_j<Q/2$ for each $j\in [2(k-1)]$.

For any $0<\epsilon\leq 1$, define $$\lambda=\max\left\{600k^2;\left\lceil\frac{4(k+1)}{\epsilon}+\frac{(5k+3)(2-\epsilon)Q}{\epsilon \overline{q}}\right\rceil-2\right\}.$$ We create an instance $I'$ of the {\sf RMPSD} as done in the proof of Theorem \ref{maininapprox} with the addition of a buyer $0$, with $v_0=\frac{(\lambda-2k)\overline{q}}{(Q+\overline{q})k}$ and $d_0=k$, and $k+1$ items of quality $Q+\overline{q}$.

We first show that $v_0>2=v_1$. It holds
\begin{eqnarray*}
v_0 & = & \frac{(\lambda-2k)\overline{q}}{(Q+\overline{q})k}\\
& > & \frac{(\lambda-2k)\overline{q}}{(3(k-1)q_{min}+\overline{q})k}\\
& = & \frac{(\lambda-2k)\overline{q}}{(300(k-1)\overline{q}+\overline{q})k}\\
& \geq & \frac{600k^2-2k}{300k^2-299k}\\
& > & 2,
\end{eqnarray*}
where the first inequality follows from $Q\leq 2(k-1)q_{max}\leq 3(k-1)q_{min}$.

Moreover, note that, in the proof of Theorem \ref{maininapprox}, we
only needed $\lambda>3k$ in order to show that $v_i>v_{i+1}$ for
each $i\in [4]$. Hence, we can conclude that $v_i>v_{i+1}$ for each
$0\leq i\leq 4$. It follows that, with the addition of buyer $0$ and
the $k+1$ items of quality $Q+\overline{q}$, the instance $I'$ is
now proper.

The spirit of the proof is the same of that used in the one of
Theorem \ref{maininapprox}, i.e., we show that, if $I$ admits a
positive answer, then there exists a solution for $I'$ with revenue
above a certain value, while, if $I$ admits no positive answers,
then all the solutions for $I'$ must raise a revenue below a certain
other value.

First of all, let us determine the set of all possible non-empty
allocation vectors able to yield an envy-free outcome. To this aim,
we can claim the following set of constraints which come from the
fact that $v_i>v_{i+1}$ for each $0\leq i\leq 4$:
\begin{enumerate}
\item[$i)$] Since $d_0\leq d_i$ for each $i\geq 1$, it must be $X_0\neq\emptyset$;
\item[$ii)$] Since $d_1\leq d_i$ for each $i\geq 2$, it must be $X_1\neq\emptyset$ when $\bigcup_{i=2}^5 X_i\neq\emptyset$;
\item[$iii)$] Since $d_3\leq d_i$ for each $i\geq 4$, it must be $X_3\neq\emptyset$ when $X_4\cup X_5\neq\emptyset$;
\item[$iv)$] Since $d_2\leq d_3+d_4$, it must be $X_2\neq\emptyset$ when $X_3,X_4\neq\emptyset$;
\end{enumerate}
Hence, for each envy-free outcome $({\bf X},{\bf p})$, $\bf X$ can only fall into one of the following five cases:
\begin{enumerate}
\item $X_0\neq\emptyset$ and $X_i=\emptyset$ for each $i\geq 1$;
\item $X_0,X_1\neq\emptyset$ and $X_i=\emptyset$ for each $i\geq 2$;
\item $X_0,X_1,X_2\neq\emptyset$ and $X_i=\emptyset$ for each $i\geq 3$;
\item $X_0,X_1,X_3\neq\emptyset$ and $X_2,X_4,X_5=\emptyset$;
\item $X_0,X_1,X_3,X_5\neq\emptyset$ and $X_2,X_4=\emptyset$.
\end{enumerate}
When $\bf X$ falls into case $(1)$, for any pricing vector $\bf p$
such that $({\bf X},{\bf p})$ is envy-free, it holds $rev({\bf
X},{\bf p})\leq v_0 k (Q+\overline{q})=(\lambda-2k)\overline{q}$.
When $\bf X$ falls into case $(2)$, for any pricing vector $\bf p$
such that $({\bf X},{\bf p})$ is envy-free, it holds $rev({\bf
X},{\bf p})\leq v_0 k
(Q+\overline{q})+2(kQ+\overline{q})=(\lambda-2k)\overline{q}+2(kQ+\overline{q})$.
When $\bf X$ falls into case $(4)$, for any pricing vector $\bf p$
such that $({\bf X},{\bf p})$ is envy-free, it holds $rev({\bf
X},{\bf p})\leq v_0 k (Q+\overline{q})+2(kQ+\overline{q})+\frac 3 2
v_3 kQ<(\lambda-2k)\overline{q}+5kQ+2\overline{q}$ since $v_3<2$.

When $\bf X$ falls into case $(3)$, $X_0$ can only contain items of
quality $Q+\overline{q}$, the remaining item of quality
$Q+\overline{q}$, denote it by $j$, must be assigned to $X_1$ and
$X_2$ must contain an item of quality $Q$. For any pricing vector
$\bf p$ such that $({\bf X},{\bf p})$ is envy-free, there must exist
an item $j'\in X_2$ such that $p_{j'}\leq v_2 q_{j'}<2q_{j'}$.
Moreover, it must be $u_{1j}=2(Q+\overline{q})-p_j\geq
u_{1j'}=2q_{j'}-p_{j'}>0$ which implies $p_{j}\leq
2(Q+\overline{q})$. Finally, for each item $j''\in X_0$, it must be
$p_{j''}=p_{j'}$ since $q_{j''}=q_{j'}$. Hence, it holds
\begin{eqnarray*}
rev({\bf X},{\bf p}) & \leq & 4kQ+2(k+1)\overline{q}+v_2\left(\frac 5 2 Q+(\lambda-2k)\overline{q}\right)\\
& = & 4kQ+2(k+1)\overline{q}+\frac 5 2 Q+(\lambda-2k)\overline{q}\\
& & +\frac 1 \lambda \frac{Q-2k\overline{q}+kQ(\lambda+1)/2}{Qk+Q-2k\overline{q}+\lambda\overline{q}} \left(\frac 5 2 Q+(\lambda-2k)\overline{q}\right)\\
& = & \left(4k+\frac 5 2\right)Q+(\lambda+2)\overline{q}+\frac{(2(\lambda-2k)\overline{q}+5Q)(kQ(\lambda+1)+2Q-4k\overline{q})}{4\lambda((\lambda-2k)\overline{q}+(k+1)Q)}\\
& < & (4k+3)Q+(\lambda+2)\overline{q}+\frac{kQ(\lambda+1)+2Q}{2\lambda}\\
& < & (4k+3)Q+(\lambda+2)\overline{q}+kQ\\
& = & (5k+3)Q+(\lambda+2)\overline{q},
\end{eqnarray*}
where the first strict inequality follows from $2(k+1)>5$ and the second one follows from $k+2<k\lambda$.

Hence, we can conclude that, when $\bf X$ falls into one of the
cases from $(1)$ to $(4)$, for any pricing vector $\bf p$ such that
$({\bf X},{\bf p})$ is envy-free, it holds $rev({\bf X},{\bf
p})<(5k+3)Q+(\lambda+2)\overline{q}$.

In the remaining of this proof, we restrict to the case in which $\bf X$ falls into case $(5)$.

\begin{lemma}\label{lemma1bis}
If there exists a positive answer to $I$, then there exists an
envy-free outcome for $I'$ of revenue greater than
$2(\lambda-2k)\overline{q}$.
\end{lemma}

\begin{proof}
Consider the allocation vector ${\bf X}$ such that $X_0$ contains
the $k$ items of quality $Q+\overline{q}$, $X_1$ contains $k$ items
of quality $Q$, $X_3$ contains the item of quality $Q/2$ plus the
$k-1$ items forming a positive answer to $I$, $X_5$ contains the
$\lambda-2k$ items of quality $\overline{q}$ and
$X_2=X_4=\emptyset$. Note that $\bf X$ is monotone. We show that the
outcome $({\bf X},\widetilde{{\bf p}})$ is envy-free.

According to the price vector $\widetilde{\bf p}$, it holds
$\widetilde{p}_j=\frac{(\lambda-2k)\overline{q}}{k}+\left(3+\frac 1
\lambda\right)\frac Q 2-\frac{\overline{q}}{\lambda}$ for each $j\in
X_0$, $\widetilde{p}_j=\frac{(3\lambda+1)Q-2\overline{q}}{2\lambda}$
for each $j\in X_1$,
$\widetilde{p}_j=\frac{(\lambda+1)q_j-\overline{q}}{\lambda}$ for
each $j\in X_3$ and $\widetilde{p}_j=\overline{q}$ for each $j\in
X_5$.

Because of Lemma \ref{prezzatura}, in order to show that $({\bf
X},\widetilde{{\bf p}})$ is envy-free, we only need to prove that,
for each buyer $i\notin W({\bf X})$ and $T\subseteq M$ with
$|T|=d_i$, it holds $\sum_{j\in T}u_{ij}\leq 0$. Note that the
buyers not belonging to $W({\bf X})$ are buyers $2$ and $4$.

For buyer $2$, for each pair of items $(j,j')$ with $j\in X_0$ and
$j'\in X_1$, it holds $u_{2j}<u_{2j'}$, for each pair of items
$(j',j'')$ with $j'\in X_1$ and $j''\in X_3$, it holds
$u_{2j'}<u_{2j''}$ and, for each pair of items $(j',j''')$ with
$j'\in X_1$ and $j'''\in X_5$, it holds $u_{2j'}<u_{2j'''}$. In
fact, we have
\begin{eqnarray*}
u_{2j'}-u_{2j} & = & \frac{\lambda\overline{q}}{k} -2\overline{q}-v_2 \overline{q}\\
& > & \left(\frac \lambda k -4\right)\overline{q}\\
& > & 0,
\end{eqnarray*}
where the first inequality follows from $v_2<2$ and the second one follows from $\lambda>4k$;
\begin{eqnarray*}
u_{2j''}-u_{2j'} & = & v_2 q_j -q_j-\frac{q_j}{\lambda}-v_2 Q+\frac 3 2 Q+\frac{Q}{2\lambda}\\
& > & \frac 1 \lambda\left(\frac Q 2-q_j\right)\\
& > & 0,
\end{eqnarray*}
where the first inequality follows from $1<v_2<3/2$ and the second one follows from $q_j<Q/2$ for each $j\in X_3$; and
\begin{eqnarray*}
u_{2j'''}-u_{2j'} & = & v_2 \overline{q} -\overline{q}-v_2 Q+\frac 3 2 Q+\frac{Q}{2\lambda}-\frac{\overline{q}}{\lambda}\\
& > & 0,
\end{eqnarray*}
where the inequality follows from $1<v_2<3/2$ and $\overline{q}<Q/2$.

Hence, the set of items of cardinality $d_2$ which gives the highest utility to buyer $2$ is $T=X_1\cup X_3\cup X_5$. It holds
\begin{displaymath}
\begin{array}{ll}
& \displaystyle\sum_{j\in T}\left(v_2 q_j-\widetilde{p}_j\right)\\
= & k\left(v_2 Q-\frac 3 2 Q-\frac{Q}{2\lambda}+\frac{\overline{q}}{\lambda}\right)+v_2 Q-Q-\frac Q \lambda+\frac{k\overline{q}}{\lambda}+(\lambda-2k)(v_2 \overline{q}-\overline{q})\\
= & 0.
\end{array}
\end{displaymath}

For buyer $4$, for each pair of items $(j,j')$ with $j\in X_0$ and
$j'\in X_1$, it holds $u_{4j}<u_{4j'}$, for each pair of items
$(j',j'')$ with $j'\in X_1$ and $j''\in X_3$, it holds
$u_{4j'}<u_{4j''}$ and, for each pair of items $(j',j''')$ with
$j'\in X_1$ and $j'''\in X_5$, it holds $u_{4j'}<u_{4j'''}$. In
fact, we have
\begin{eqnarray*}
u_{4j'}-u_{4j} & = & \frac{\lambda\overline{q}}{k} -2\overline{q}-v_4 \overline{q}\\
& > & \left(\frac \lambda k -4\right)\overline{q}\\
& > & 0,
\end{eqnarray*}
where the first inequality follows from $v_4<2$ and the second one follows from $\lambda>4k$;
\begin{eqnarray*}
u_{4j''}-u_{4j'} & = & v_4 q_j -q_j-\frac{q_j}{\lambda}-v_4 Q+\frac 3 2 Q+\frac{Q}{2\lambda}\\
& > & \frac 1 \lambda\left(\frac Q 2-q_j\right)\\
& > & 0,
\end{eqnarray*}
where the first inequality follows from $1<v_4<3/2$ and the second one follows from $q_j<Q/2$ for each $j\in X_3$; and
\begin{eqnarray*}
u_{4j'''}-u_{4j''} & = & v_4 \overline{q} -\overline{q}-v_4 q_j+q_j+\frac{q_j}{\lambda}-\frac{\overline{q}}{\lambda}\\
& = & (q_j-\overline{q})\left(1+\frac 1 \lambda-v_4\right)\\
& > & 0
\end{eqnarray*}
where the inequality follows from $v_4<1+1/\lambda$ and $q_j>\overline{q}$ for each $j\in X_3$.

Hence, the set of items of cardinality $d_4$ which gives the highest utility to buyer $4$ is $T=X_3\cup X_5$. It holds
\begin{displaymath}
\begin{array}{ll}
& \displaystyle\sum_{j\in T}\left(v_4 q_j-\widetilde{p}_j\right)\\
= & v_4 Q-Q-\frac Q \lambda+\frac{k\overline{q}}{\lambda}+(\lambda-2k)(v_4 \overline{q}-\overline{q})\\
= & 0.
\end{array}
\end{displaymath}

Hence, we can conclude that the outcome $({\bf X},\widetilde{{\bf
p}})$ is envy-free and it holds $rev({\bf X},\widetilde{{\bf
p}})>2(\lambda-2k)\overline{q}$.\qed
\end{proof}

We continue by showing that, for any envy-free outcome $({\bf
X},{\bf p})$ falling into case $(5)$ and such that $X_1$ contains an
item of quality $Q+\overline{q}$, it holds $rev({\bf X},{\bf
p})<(\lambda+2)\overline{q}+(4k+3)Q$.

Note that, in such a case, by Lemma \ref{orderedqualities}, $X_0$
can only contain items of quality $Q+\overline{q}$. For any pricing
vector $\bf p$ such that $({\bf X},{\bf p})$ is envy-free, there
must exist an item $j'\in X_5$ such that $p_{j'}\leq q_{j'}$. Let
$j''$ be the index of the item of quality $Q+\overline{q}$ belonging
to $X_1$. By the envy-freeness of $({\bf X},{\bf p})$, it holds
$u_{1j''}=2(Q+\overline{q})-p_{j''}\geq 2q_{j'}-p_{j'}=q_{j'}$ which
implies $p_{j''}<2(Q+\overline{q})$. Clearly, since $({\bf X},{\bf
p})$ is envy-free, for each item $j\in X_0$, it must be
$p_j=p_{j''}$ since $q_j=q_{j''}$. Hence, it holds $rev({\bf X},{\bf
p})<4kQ+2(k+1)\overline{q}+\frac 3 2 v_3
Q+(\lambda-2k)\overline{q}<(\lambda+2)\overline{q}+(4k+3)Q$ because
$v_3<2$.

Since it holds
$(\lambda+2)\overline{q}+(4k+3)Q<(\lambda+2)\overline{q}+(5k+3)Q$,
it follows that, either when $\bf X$ falls into case $(5)$ and $X_1$
contains an item of quality $Q+\overline{q}$ or $\bf X$ falls into
one of the cases from $(1)$ to $(4)$, for any pricing vector $\bf p$
such that $({\bf X},{\bf p})$ is envy-free, it holds $rev({\bf
X},{\bf p})\leq (\lambda+2)\overline{q}+(5k+3)Q$.

Now we are only left to consider envy-free outcomes $({\bf X},{\bf
p})$ such that $\bf X$ falls into case $(5)$ and $X_1$ does not
contain any item of quality $Q+\overline{q}$.

Assume that $\sum_{j\in X_3}>Q$. This can only happen when buyer $3$
is assigned an item of quality at least $Q/2$. In such a case, since
$X_1$ does not contain any item of quality $Q+\overline{q}$, it can
only be the case that each item in $X_1$ is of quality $Q$ and $X_3$
gets the item of quality $Q/2$. This means that the items allocated
by ${\bf X}$ to buyers $1$, $3$ and $5$ are drawn from the same
instance $I'$ considered in the proof of Theorem \ref{maininapprox}.
Hence, we can replicate the arguments used in the proof of Lemma
\ref{lemma2} to show that $\sum_{j\in X_3}>Q$ yields a
contradiction.

Similarly, assume that $\sum_{j\in X_3}<Q$. This can only happen
when the items allocated by ${\bf X}$ to buyers $3$ and $5$ are
drawn from the same instance $I'$ considered in the proof of Theorem
\ref{maininapprox}. Hence, we can replicate the arguments used in
the proof of Lemma \ref{lemma3} to show that $\sum_{j\in X_3}<Q$
yields a contradiction.

We can conclude that there exists an envy-free outcome $({\bf
X},{\bf p})$ falling into case $(5)$ in which no item of quality
$Q+\overline{q}$ belongs to $X_1$ only if $\sum_{j\in X_3}q_j=Q$.
Since, as we have already observed, in such a case the item of
quality $Q/2$ has to belong to $X_3$, it follows that there exists
an envy-free outcome $({\bf X},{\bf p})$ falling into case $(5)$ in
which no item of quality $Q+\overline{q}$ belongs to $X_1$ only if
there are $k-1$ items inherited from $I$ whose sum is exactly $Q/2$,
that is, only if $I$ admits a positive solution.

Any other envy-free outcome can raise a revenue of at most
$(\lambda+2)\overline{q}+(5k+3)Q$. Hence, if there exists a positive
answer to $I$, then, by Lemma \ref{lemma1bis}, there exists a
solution to $I'$ of revenue strictly greater than
$2(\lambda-2k)\overline{q}$, while, if there is no positive answer
to $I$, then there exists no solution to $I'$ of revenue more than
$(\lambda+2)\overline{q}+(5k+3)Q$.

Thus, if there exists an $r$-approximation algorithm for the {\sf
RMPSD} on continuous instances with
$r\leq\frac{2(\lambda-2k)\overline{q}}{(\lambda+2)\overline{q}+(5k+3)Q}$,
it is then possible to decide in polynomial time the {\sf
Constrained Partition} problem, thus implying ${\sf P}={\sf NP}$. By
$\lambda\geq\frac{4(k+1)}{\epsilon}+\frac{(5k+3)(2-\epsilon)Q}{\epsilon
\overline{q}}-2$, it follows
$\frac{2(\lambda-2k)\overline{q}}{(\lambda+2)\overline{q}+(5k+3)Q}\geq
2-\epsilon$ which implies the claim.\qed
\end{proof}

\end{document}